\documentclass[10pt,twocolumn]{article}
\usepackage{cite}
\usepackage{amsmath,amssymb,amsfonts,amsthm}
\usepackage{multicol}
\usepackage[table,xcdraw]{xcolor}
\usepackage{caption}
\usepackage{graphicx}
\graphicspath{{./img/}}
\usepackage{csquotes}
\usepackage{todonotes}
\usepackage{subfigure}
\usepackage{csquotes}
\usepackage{lipsum}

\def\BibTeX{{\rm B\kern-.05em{\sc i\kern-.025em b}\kern-.08em
    T\kern-.1667em\lower.7ex\hbox{E}\kern-.125emX}}
\usepackage{hyperref}
\usepackage[margin=1.75cm]{geometry}
\usepackage[hpos=0.17\paperwidth,vpos=0.97\paperheight,angle=0,scale=0.3]{draftwatermark}
\SetWatermarkText{Preprint}
\SetWatermarkLightness{0.5}
\SetWatermarkText{Open Access License: CC-BY-SA 4.0}

\newtheorem{definition}{Definition}[section]
\newtheorem{theorem}{Theorem}[section]
\newtheorem{corollary}{Corollary}[theorem]
\newtheorem{example}{Example}[theorem]
\newtheorem{lemma}[theorem]{Lemma}

\newcommand{\longhookrightarrow}{\lhook\joinrel\longrightarrow}

\DeclareMathOperator{\EE}{\mathbb{E}}

\title{A Mathematical Theory of Payment Channel Networks}

\author{Rene Pickhardt\thanks{Due to personal circumstances related to my mental health I have not been able to receive a PhD - despite being a researcher since 2011. If you can give someone a PhD and think this work and my prior work about reliability in payment channel networks since 2018 reaches the merits of a PhD, then please feel free to get in contact with me. Thank you!} (r.pickhardt@gmail.com)}
\begin{document} 
\maketitle

\begin{abstract}
We introduce a geometric theory of payment channel networks that centers the polytope $W_G$ of feasible wealth distributions, with liquidity states $L_G$ naturally projecting onto $W_G$ via strict circulations. A payment is feasible iff the post-transfer wealth remains in $W_G$. As every infeasible payment needs at least one on-chain transaction to change the topology of the network this yields a simple throughput law linking off-chain scale to feasibility: If $\zeta$ is on-chain settlement bandwidth and $\rho$ the expected fraction of infeasible payments, the sustainable off-chain bandwidth satisfies
\[
\mathcal{S}=\frac{\zeta}{\rho}
\]

We demonstrate for every $S\subset V$, the total wealth $\omega(S)$ of $S$ must lie in an interval whose width equals the cut capacity $C(\delta(S))$. This is used to quantify how multi-party channels (aka coinpools / channel factories) expand $W_G$: Modeling a $k$-party channel as a $k$-uniform hyperedge strictly widens every cut in expectation, so $W_G$ grows monotonically in $k$. In particular, a single node’s expected accessible wealth scales linearly with $k/n$.

Separately, we analyze depletion: Under linear (asymmetric) fees, cost-minimizing flow within a wealth fiber pushes cycles to the boundary, generically depleting channels except for a residual spanning forest. The theory suggests three mitigation levers: (i) \textbf{symmetric fees} per direction (simple but likely unpopular) (ii) \textbf{convex/tiered fees} (local flow control but at odds with source routing without broadcasting realtime liquidity information), and (iii) \textbf{coordinated replenishment} (choose an optimal circulation within a fiber that is closest to the wishes of the participants).

The theory in this paper explains why two-party channel networks will struggle to scale off-chain payment bandwidth and why multi-party primitives are structurally more capital-efficient and yield higher expected payment bandwidth. We furthermore show how subtle protocool choices like the structure of the fee function or pricing and coordination keep operation inside the feasible region, improving reliability of feasible payments.
\end{abstract}

\section{Introduction}
The Bitcoin protocol faces significant limitations, notably its support for only about seven transactions per second and the requirement for users to wait several minutes on average for transaction confirmation. To address these issues, the Lightning Network was developed as a second-layer protocol. Barring exploitation of known DoS attacks \cite{harris2020flood,shikhelman2022unjamming,tochner2020route}, the Lightning Network allows users to conduct near real-time Bitcoin payments. Additionally, it technically supports a higher payment throughput compared to the Bitcoin protocol.

However, these enhancements come at the cost of reduced payment reliability on the Lightning Network, particularly as the desired payment amounts increase. To quantify this reduced reliability, we examine the geometry of two-party payment channel networks and their generalization to \( k \)-party payment channel networks. By analyzing the high-dimensional geometric properties of these networks, we explain the emergence of two phenomena frequently observed by Lightning Network participants, often associated with reduced payment reliability:

\begin{enumerate}
\item \textbf{Channel Depletion}: Most of the liquidity in payment channels is likely to be controlled by one of the two peers maintaining the channel.
\item \textbf{Infeasible Payments}: Even with numerous payment attempts and optimal routing techniques, a fraction of payments cannot be fulfilled.
\end{enumerate}

Node operators often cite two primary reasons for channel depletion. First, the lack of a circular economy among network users results in the existence of source and sink nodes, leading to natural channel drains \cite{pickhardt2022valves}. Second, it is believed that node operators have not yet optimized routing fee settings \cite{zhang2023rethinking}, which affect senders' payment route selections. Routing fees are thought to be useful for better flow control and channel balance \cite{ren2018optimal}. Payment failure rates are often attributed to channel depletion \cite{alscher2023price}, suboptimal payment routing strategies, liquidity management issues within channels \cite{pickhardt2020imbalance}, and users' incomplete information about the network's liquidity state \cite{pickhardt2021security}.

We observe a contradiction between these reasons for channel depletion. While routing fees impact flow control, they cannot alter the distribution of payment requests among network participants, which leads to the emergence of sources and sinks. Regarding payment failure rates, the uncertainty about channel states has been addressed \cite{pickhardt2021security}, and optimal methods for making routing decisions despite incomplete information are known \cite{pickhardt2021optimally}. Probability density functions for the depleted channel model also exist \cite{bitromortac2024blazing,rossi2024channel}.

\section{Review of some Graph Theory to describe Payment Channel Networks}
%% We start by reviewing the standard mathematical model to describe the Lightning Network.
%% For simplicity of our observations and the presentation we ignore routing fees and other meta data of channels for the rest of this document.
%% The Lightning Network is usually modelled as an undirected weighted multigraph.
%% However for our considerations multi edges could be combined into one bigger edge.
%% Thus we start from a weighted and undirected graph $G(V,E,cap)$.
%% $V$ corresponds to the set of $n=|V|$ vertices (aka peers) of the network.
%% $E$ contains $m = |E|$ elements of $V\times V$.
%% We call them the $m$ edges (aka payment channels).
%% The weighting of the edges corresponds to the \textbf{capacity} of the channel via the function $cap: E\longrightarrow\mathbb{N}$.
%% We write $c_e = cap(e)$ and  call $C = \sum_{e\in E} c_e$ the total capacity of the network.
%% This is also called the total liquidity or the number of coins in the payment channel network.

%% The protocoll allows to specify how many coins  are owned by each of the peers within each channel.
%% \begin{definition}
%%  Let $e=(u,v)\in E$ we call $e_u,e_v \in\{0,\dots,cap(e)\}$ the liquidity that $u$ and $v$ respectively owns in the channel $e$.
%% \end{definition}

%% On the entire network we can encode the assignment of the liquidity in each channel (or the network state) to its peers via the following liquidity function:
%% %The liquidity $c_i$ is mutually owned by the nodes $u$ and $v$.
We begin by reviewing the standard mathematical model used to describe the Lightning Network.
For simplicity, we ignore routing fees and other metadata of channels throughout this document.
The Lightning Network is typically modeled as an undirected weighted multigraph.
However, for our considerations, multi-edges can be combined into one larger edge.
Thus, we start with a weighted and undirected graph \(G(V,E,cap)\).
\(V\) corresponds to the set of \(n = |V|\) vertices (also known as peers) in the network.
\(E\) contains \(m = |E|\) elements of \(V \times V\).
These are the \(m\) edges (also known as payment channels).
The weighting of the edges corresponds to the \textbf{capacity} of the channel via the function \(cap: E \longrightarrow \mathbb{N}\).
We write \(c_e = cap(e)\) and call \(C = \sum_{e \in E} c_e\) the total capacity of the network.
This is also called the total liquidity or the number of coins in the payment channel network.

The protocol specifies how many coins are owned by each of the peers within each channel.
\begin{definition}
  Let \(e = (u, v) \in E\). 
  We call \(e_u, e_v \in \{0, \dots, cap(e)\}\) the liquidity that \(u\) and \(v\) respectively own in the channel \(e\).
\end{definition}

%%We can encode the assignment of the liquidity in each channel (or the network state) to its peers via the following liquidity function:
We can represent the assignment of liquidity in each channel (or the network state) to its peers using the following liquidity function:
\begin{definition}
  \label{def:liquidityFunction}
  The liquidity function of the network $\lambda: E\times V \longrightarrow \{0,\dots,C\}$ is defined through:
    \begin{equation}
    \lambda(e,x)= \begin{cases}
      e_x & \text{if $x\in e$}\\
      0 & \text{otherwise}
    \end{cases}
    \end{equation}
\end{definition}

%% By the design and properties of payment channels the channel states and thus the $2\cdot m$ variables $e_x$ wiht $e\in E, x\in e$ are not independent of each other.
%% The constraints to the liquidity function is so important for the remainder of this document that we give it its own name: 

%% \begin{definition}
%%   The principle of \textbf{conservation of liquidity} encodes the constraints that the liquidity function has on the states of the channels.  
%% \begin{equation}
%%   \label{eq:conservationOfLiquidity}
%%   \begin{split}
%%   \lambda(e,u) + \lambda(e,v) & = cap(e) \\
%%   \Leftrightarrow e_u + e_v & = c_e
%%   \end{split}
%% \end{equation}

%% \end{definition}
Due to the design and properties of payment channels the channel states, and thus the \(2 \cdot m\) variables \(e_x\) with \(e \in E\) and \(x \in e\), are not independent of each other.
The constraints on the liquidity function are so crucial for the remainder of this document that we assign them a specific term: 

\begin{definition}
  The constraints that the protocol imposes on the liquidity function
\begin{equation}
  \label{eq:conservationOfLiquidity}
  \begin{split}
  \lambda(e,u) + \lambda(e,v) & = cap(e) \\
  \Leftrightarrow e_u + e_v & = c_e
  \end{split}
\end{equation}
is called \textbf{conservation of liquidity}.
\end{definition}

%% Conversation of liquidity refers to the fact that coins that are attached to a channel cannot leave it without conducting an on-chain transaction.
%% Despite this, the ownership of coins may be routed as a network flow from one peer to another one through the network of channels.
%% However not all payments are feasible as the topology and liquidity function produce a natural constraint on the feasible flows.

%% \begin{definition}
%% Given a fixed network $G(V,E,cap)$ and a fixed liquidity state von $G$ via an liquidity function $\lambda$ we can create the associated liquidity network $\mathcal{L}(G,\lambda)$ 
%% $\mathcal{L}(G,\lambda)$ is a directed flow network $(V',E',c)$ with capacity function $c:E'\longrightarrow\mathbb{N}$.
%% First we set $V'= V$.
%% Furthermore for every $e=(u,v)\in E$ we add two edges $(u,v)$ and $(v,u)$ to $E´$ such that $c(u,v)=\lambda(e,u)$ and $c(v,u)=\lambda(e,v)$.
%% \end{definition}

%% The Liquidity Network $\mathcal{L}(G,\lambda)$ thus encodes how money can flow through the network $G$ in a given state $\lambda$.
%% In fact a payment of amount $a$ between user $i$ and $j$ is feasible if and only if the minimum $ij$-cut on $\mathcal{L}(G,\lambda)$ is larger than $a$.
%% This means that deciding feasibility of a payment if the network is in a given state $\lambda$ can be achieved in almost linear time\cite{chen2022maximum,chen2023almost,van2023deterministic}

The principle of conservation of liquidity implies that coins attached to a channel cannot leave it without an on-chain transaction.
Despite this restriction, the ownership of coins can be routed through the network of channels as a network flow from one peer to another.
However, not all payments are feasible due to the constraints imposed by the network's topology and liquidity function on the feasible flows.

\begin{definition}
Given a fixed network \(G(V,E,cap)\) and a fixed liquidity state of \(G\) defined by a liquidity function \(\lambda\), we can create the associated liquidity network \(\mathcal{L}(G,\lambda)\).
\(\mathcal{L}(G,\lambda)\) is a directed flow network \((V',E',c)\) with a capacity function \(c: E' \longrightarrow \mathbb{N}\).
We set \(V' = V\).
For every \(e = (u,v) \in E\), we add two directed edges \((u,v)\) and \((v,u)\) to \(E'\) such that \(c(u,v) = \lambda(e,u)\) and \(c(v,u) = \lambda(e,v)\).
\end{definition}

The liquidity network \(\mathcal{L}(G,\lambda)\) thus encodes the possible flows of coins through the network \(G\) in a given state \(\lambda\).
A payment of amount \(a\) between users \(i\) and \(j\) is feasible if and only if the minimum \(ij\)-cut in \(\mathcal{L}(G,\lambda)\) is greater than \(a\).
Determining the feasibility of a payment in a given state \(\lambda\) can be achieved in almost linear time \cite{chen2022maximum,chen2023almost,van2023deterministic}

%TODO: Include this here or add this to the evaluation! We have conducted a study about global min cut distributions (assuming network states were equally likeli (which is a poor choice) Nice! with the current version I don´t need to assume a bimodal liquidity split but I can start from uniformly distributed wealth distributions and compute the min cut distribution!)
%Breaking the Cubic Barrier for All-Pairs Max-Flow: Gomory-Hu Tree in Nearly Quadratic Time
%https://ieeexplore.ieee.org/abstract/document/9996943?casa_token=EaoUMdG10ecAAAAA:AxeqEWb9X14gC3jYNGNV4uLKzPyxycqrqAI7Zj-G-t7zItxLbr5V1d76CMyiqOXj8gXHcrzZwAtAXd8

%TODO: maybe better name: Polytope of Liquidity states (instead of channel states)?
\section{Polytope of Liquidity States of a Payment Channel Network}
\label{sec:statePolytope}

%% The liquidity function is unknown to the peers of the network.
%% They may have partial knowledge of the function on a region around their own position within the network.
%% For example a peer $v\in V$ is aware of the liquidity in her channels.
%% Instead of studying or guessting a particular liquidity function we address this lack of information by studying the set of feasible liquidity functions and the geometry of that set
The liquidity function is generally unknown to the peers within the network.
Peers may have partial knowledge of the liquidity function in the region surrounding their own position within the network.
For instance, a peer \(v \in V\) is aware of the liquidity in their own channels.
Also peers may learn about the current liquidity function while attempting to deliver payments or through actively probing the network\cite{tikhomirov2020probing}.
Rather than focusing on or estimating a specific liquidity function, we address this uncertainty by examining the set of feasible liquidity functions and the geometry of that set.

\begin{definition}
For a fixed payment channel network \( G(V,E, \text{cap}) \), we define the set of all feasible liquidity functions, or the set of liquidity states of the network, as:
\[ L_G = \{\lambda: E \times V \longrightarrow \{0, \dots, C\} \mid \lambda(e,u) + \lambda(e,v) = \text{cap}(e) \}. \]
\end{definition}
This definition motivates the use of integer linear programming to examine \(L_G\) because the constraints are linear equations.
\begin{lemma}
The set \( L_G \) of all feasible liquidity states can be embedded into \(\mathbb{Z}^{2m}\).
\end{lemma}

\begin{proof}
Consider a basis \( \beta = \{ b_{e,x} \in \mathbb{Z}^{2m} \mid e \in E, x \in e \} \) where \( b_{e,x} \) are the \(2m\) basis vectors that span \(\mathbb{Z}^{2m}\). There is a natural embedding:
\[
\begin{aligned}
\iota: L_G & \longhookrightarrow \mathbb{Z}^{2m}, \\
\lambda & \stackrel{\iota}{\longmapsto} \sum_{e \in E, x \in e} e_x \cdot b_{e,x} \\
& = \sum_{e=(u,v) \in E} \left( e_u \cdot b_{e,u} + e_v \cdot b_{e,v} \right).
\end{aligned}
\]
\end{proof}

Let \( e_i = (u_i, v_i) \) be the \( i \)-th edge, then we can fix an order of \( \beta \) via:
\( \beta = \langle b_{e_1, u_1}, b_{e_1, v_1}, \dots, b_{e_m, u_m}, b_{e_m, v_m} \rangle. \)
In particular, we write:\( c_i = \text{cap}(e_i). \)

\begin{lemma}
\label{lem:dimensionOfStatePolytope}
The set \( L_G \) of feasible liquidity functions on a network \( G \) is isomorphic to an \( m \)-dimensional hyperbox \( H_G \subset \mathbb{Z}^{m} \). In particular:
\[ H_G = \{ 0, \dots, c_1 \} \times \dots \times \{ 0, \dots, c_m \}. \]
\end{lemma}
\begin{proof}
We prove this by showing that the forgetful function \( f: \iota(L_G) \longrightarrow H_G \) defined as
\[ \iota(\lambda) = \sum_{e=(u,v) \in E} \left( e_u \cdot b_{e,u} + e_v \cdot b_{e,v} \right) \stackrel{f}{\longmapsto} \sum_{e=(u,v) \in E} e_u \cdot b_{e,u} \]
is bijective. This involves showing that \( f \) is both injective and surjective.

\textbf{Injectivity}: Let \( \lambda \neq \mu \in L_G \). There must be at least one edge \( e=(u,v) \in E \) such that \( \lambda(e,u) \neq \mu(e,u) \). Consequently, by definition of \( f \), we have \( f(\iota(\lambda)) \neq f(\iota(\mu)) \). This is sufficient to show that \( f \) is injective.

\textbf{Surjectivity}: To show that \( f \) is surjective, take an arbitrary point \( x = (x_1, \dots, x_m) \in H \). We choose \( \lambda \) such that for all \( e_i \in E \):
\( \lambda(e_i, u_i) = x_i. \)
Because of the conservation of liquidity, we have:
\( \lambda(e_i, v_i) = \text{cap}(e_i) - \lambda(e_i, u_i) = c_i - x_i. \)
Since \( x_i \in \{0, \dots, c_i\} \), we can conclude that both \( \lambda(e_i, u_i) = x_i \) and \( \lambda(e_i, v_i) = c_i - x_i \) take values from \(\{0, \dots, c_i\}\). Thus, we have constructed \( \lambda \) such that \( f(\lambda) = x \), showing that \( f \) is also surjective.

This concludes our proof that \( f \) is bijective and that \( L_G \) is isomorphic to the \( m \)-dimensional hyperbox \( H_G \).
\end{proof}

We can see a low dimensional visualization of the introduced concepts in figure \ref{fig:statePolytopeExample}.

\begin{figure}[h]
  \centering
  \subfigure[]{\includegraphics[width=0.38\textwidth]{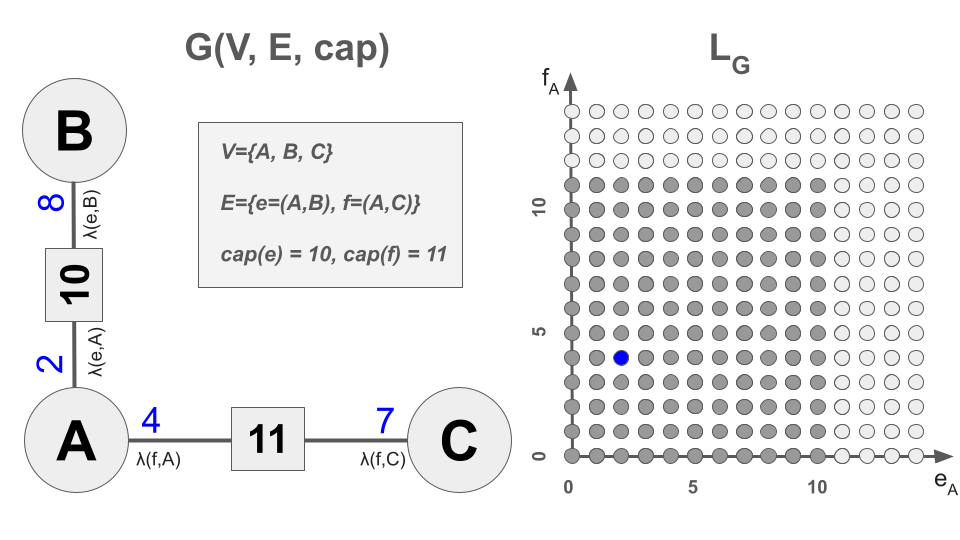}}
  \subfigure[]{\includegraphics[width=0.38\textwidth]{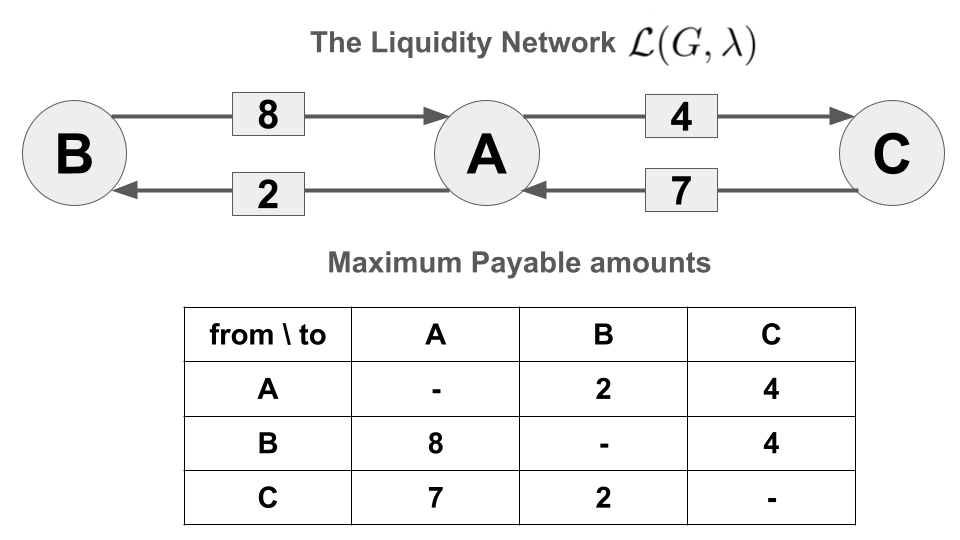}}
  \caption{An Example Network \( G \) with corresponding Polytope \( L_G \) of liquidity states and liquidity network \( \mathcal{L}(G,\lambda) \) for a fixed \( \lambda \in L_G \).}
  \label{fig:statePolytopeExample}
\end{figure}

For two-party channels, as currently deployed on the Lightning Network, knowing the liquidity of one peer is sufficient to determine the entire state of the channel. Typically, \( e_v \) is only known to the peers \( u \) and \( v \) who maintain the channel \( e \). The state \( e_v \) can change over time as the node conducts payments or fulfills routing requests from other nodes.

%https://math.stackexchange.com/questions/1517517/maximize-volume-of-hyperrectangle-given-fixed-sum-of-components
%https://math.stackexchange.com/questions/3074762/prove-than-cube-has-a-bigger-volume-than-cuboid-with-the-same-sum-of-edge-length
%% As a consequence of lemma \ref{lem:dimensionOfStatePolytope} the number of elements (aka network states) in $L_G$ is given by
%% $$vol(L_G):=|L_G| = \prod_{i=1}^m(c_i+1)$$
%% Assuming that the number of channels $m$ divides the total number of coins $C$ and as a consequence of the inequality of arithmetic and geometric means\cite{steele2004cauchy} one can see that the network in which all channels have the same capacity $\frac{C}{m}$ maximizes the volume of the corresponding polytope $L_G$ of Liquidity states.
%% In particular the maximum volume is given by $vol(S)=\left(\frac{C}{m}+1\right)^m$ and is depicted in figure \ref{fig:statePolytope}.

%% \begin{figure}[h]
%% \centering
%% \includegraphics[width=0.5\textwidth]{size_of_state_polytope}
%% \caption{Volumes $vol(L_G)=\left(\frac{C}{m}+1\right)^m$ of state polytopes with equally sized channels for various number $C$ of coins and various number $m$ of channels.}
%% \label{fig:statePolytope}
%% \end{figure}

\begin{corollary}
The number of feasible network states is given by the volume of \( L_G \) which can be computed as:
\[ \text{vol}(L_G) := |L_G| = \prod_{i=1}^m (c_i + 1). \]
\end{corollary}

Assuming the number of channels \( m \) divides the total number of coins \( C \), and using the inequality of arithmetic and geometric means \cite{steele2004cauchy}, it can be shown that the network in which all channels have the same capacity \( \frac{C}{m} \) maximizes the volume of the corresponding polytope \( L_G \) of liquidity states. 
In particular, the maximum number of feasible network states is given by:
\[ \text{vol}(L_G) = \left( \frac{C}{m} + 1 \right)^m, \]
and is depicted in figure \ref{fig:statePolytope} for various number of coins and channels.

\begin{figure}[h]
\centering
\includegraphics[width=0.5\textwidth]{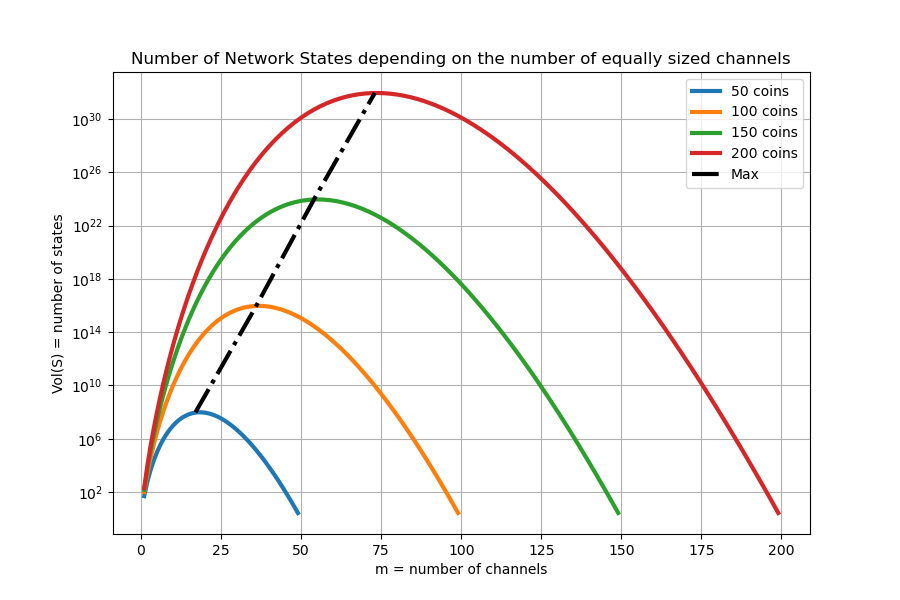}
\caption{Volumes \( \text{vol}(L_G) = \left( \frac{C}{m} + 1 \right)^m \) of state polytopes with equally sized channels for various numbers \( C \) of coins and various numbers \( m \) of channels.}
\label{fig:statePolytope}
\end{figure}

\section{Polytope of feasible Wealth Distributions}
%% So far we have focused on the liquidity states of payment channel networks.
%% From the perspective of a user we may be more interested in the wealth of the user instead of how the wealth is specifically allocated to the channels. 
%% We make a short intermezzo to study the onchain situation before showing how payment channel networks put additional constraints to the polytope of feasible wealth distributions of the users. 
%% \begin{definition}
%% Let $V = \{v_1,\dots,v_n\}$ be the set of $n$ Bitcoin users who own in total $C$ coins.
%% We call a function $\omega:V\longrightarrow\mathbb{N}_0$ a wealth distribution if and only if $C = \sum_{i=1}^n\omega(v_i)$.
%% We call
%% $$W(C,n)=\{\omega:V\longrightarrow \mathbb{N}_0 | C = \sum_{i=1}^n\omega(v_i)\}$$
%% the set of Bitcoin wealth distributions.
%% \end{definition}

%% We use the same methodology as with liquidity functions to show that $W(C,n)$ is isomorph to the Polytope $\mathcal{W}(C,n)$ that one gets by intersecting the $n$ half spaces defined as $\mathbb{Z}_i=\{x\in\mathbb{Z}^n|x_i\geq 0\}$ with the $n-1$ dimensional hyperplane $P_C$ defined by the linear equation:
%% $$P(C,n)=\{w\in\mathbb{Z}^n | w_1 + \dots + w_n = C\}$$
%% In particular:
%% $$\mathcal{W}(C,n) = P(C,n)\cap\left(\bigcap_{i=1}^n\mathbb{Z}_i\right)$$
%% In order to do so we chose a base $\beta = \{b_{v} \in \mathbb{Z}^{n}| v\in V\}$ with $b_{v}$ being the $n$ base vectors that span $\mathbb{Z}^{n}$.
%% Similar to the situation with liquidity functions there is a natural embedding:

We have focused on the liquidity states of payment channel networks so far. However, from a user's perspective, their overall wealth might be of more interest than how this wealth is allocated across different channels. Before exploring how payment channel networks impose additional constraints on the polytope of feasible wealth distributions, we will briefly examine the on-chain situation.

\begin{definition}
Let \( V = \{v_1, \dots, v_n\} \) be the set of \( n \) Bitcoin users who collectively own \( C \) coins. We define a function \( \omega: V \longrightarrow \mathbb{N}_0 \) as a wealth distribution if and only if \( C = \sum_{i=1}^n \omega(v_i) \). We denote the set of Bitcoin wealth distributions as
$$
W(C,n) = \{\omega: V \longrightarrow \mathbb{N}_0 \mid C = \sum_{i=1}^n \omega(v_i)\}.
$$
\end{definition}

We use the same methodology as with liquidity functions to show that \( W(C,n) \) is isomorphic to the polytope \( \mathcal{W}(C,n) \), which is obtained by intersecting the \( n \) half-spaces \( \mathbb{Z}_i = \{x \in \mathbb{Z}^n \mid x_i \geq 0\} \) with the \( n-1 \) dimensional hyperplane \( P_C \) defined by the linear equation:
$$
P(C,n) = \{w \in \mathbb{Z}^n \mid w_1 + \dots + w_n = C\}.
$$
Specifically,
$$
\mathcal{W}(C,n) = P(C,n) \cap \left( \bigcap_{i=1}^n \mathbb{Z}_i \right).
$$

To demonstrate this, we choose a basis \( \beta = \{b_{v} \in \mathbb{Z}^{n} \mid v \in V\} \) with \( b_{u} \) being the \( n \) base vectors that span \( \mathbb{Z}^{n} \). Similar to the case with liquidity functions, there is a natural embedding:

\begin{equation*}
  \begin{split}
\iota: W(C,n) \longhookrightarrow & \mathbb{Z}^{n} \\
\omega\stackrel{\iota}{\longmapsto} & \sum_{v\in V}\omega(v)\cdot b_{v}
\end{split}
\end{equation*}

\begin{lemma}
  The image of $\iota$ is $\mathcal{W}(C,n)$.
\end{lemma}
\begin{proof}
  Let $w\in \mathcal{W}(C,n)$ with
  $$w = \sum_{i=1}^nw_i\cdot b_{u_i}$$
  We select $\omega\in W(C,n)$ such that $\omega(v_i)=w_i$
  This is the preimage of $\iota^{-1}(w)$.
  If $w'\in \mathbb{Z}^n\backslash\mathcal{W}(C,n)$ then no $\omega$ with $\iota(\omega)=w'$ exists.
\end{proof}

%% Because of the embedding we might in future refer to wealth distributions as wealth vectors.
%% From the lemma and the stars and bars theorem\footnote{\url{https://en.wikipedia.org/wiki/Stars_and_bars_(combinatorics)}} it follows that the number $|\mathcal{W}(C,n)|$ of wealth distributions of $n$ users with a fixed number of total of coins $C$ is

%% \begin{equation}
%% vol(\mathcal{W}(C,n):=|\mathcal{W}(C,n)| = {{C + n - 1}\choose {n - 1}}
%% \end{equation}

%% For example assume the $n=3$ users Alice, Bob and Carol have $C=21$ coins that they wish to split among themselves.
%% They would have a total of $253$ ways to do so.
%% Those can be observed in the $2$-dimensional polytope:
%% \begin{figure}[h]

%% \centering
%% \includegraphics[width=0.5\textwidth]{polytope_of_bitcoin_wealth_distributions}
%% \caption{$3$-dimensional embedding of the $2$-dimensional polytope of wealth distributions of $C=21$ coins among $n=3$ users.} 
%% \label{fig:wealthPolytope}
%% \end{figure}

%% Knowing the volume of the polytope of feasible bitcoin wealth distributions will serve as reference point when studying the feasible wealth distributions $W_G$ in a payment channel network $G$. 

Because of the embedding, we may refer to wealth distributions as wealth vectors. From the lemma and the stars and bars theorem\footnote{\url{https://en.wikipedia.org/wiki/Stars_and_bars_(combinatorics)}}, it follows that the number \( |\mathcal{W}(C,n)| \) of wealth distributions for \( n \) users with a fixed total of \( C \) coins is given by:

\begin{equation}
\text{vol}(\mathcal{W}(C,n)) := |\mathcal{W}(C,n)| = \binom{C + n - 1}{n - 1}
\end{equation}

For example, assume the \( n=3 \) users Alice, Bob, and Carol have \( C=21 \) coins to distribute among themselves. They would have a total of \( 253 \) ways to do so. These distributions can be represented in a \( 2 \)-dimensional polytope:

\begin{figure}[h]
\centering
\includegraphics[width=0.5\textwidth]{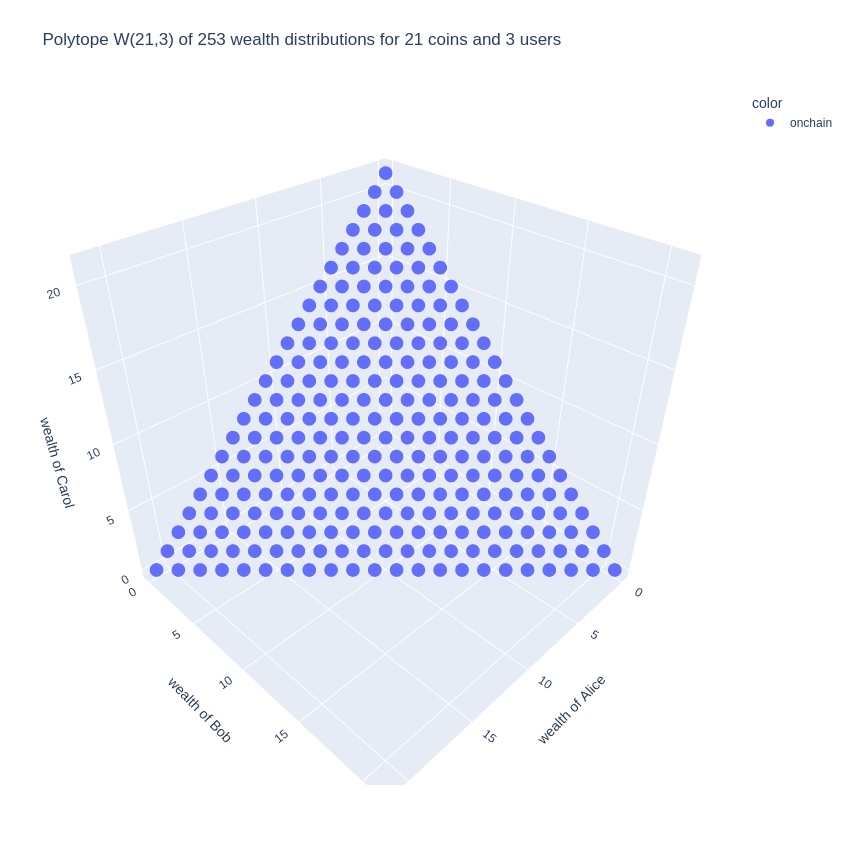}
\caption{$3$-dimensional embedding of the $2$-dimensional polytope of wealth distributions of $C=21$ coins among $n=3$ users.}
\label{fig:wealthPolytope}
\end{figure}

Knowing the volume of the polytope of feasible Bitcoin wealth distributions will serve as a reference value when studying the feasible wealth distributions \( W_G \) in a payment channel network \( G \).

\subsection{The Polytope $W_G$ of feasible Wealth Distributions of a Payment Channel Network $G$}

Instead of examining all wealth distributions of \( C \) coins among \( n \) users, we focus on those distributions that are feasible within a given payment channel network \( G(V,E,cap) \).

\begin{definition}
Given a payment channel network \( G(V,E,cap) \) with \( |V|=n \) users and \( C \) coins, a wealth distribution \( \omega: V \longrightarrow \{0,\dots,C\} \) is called feasible in \( G \) if there exists a liquidity function \( \lambda \) such that for all \( v \in V \):
\begin{equation}
\sum_{e \in E: v \in e} \lambda(e,v) = \omega(v)
\end{equation}
We denote the set of wealth distributions that are feasible in \( G \) as:
$$\{\omega \in W(C,n) | \exists \lambda \in L_G: \sum_{e \in E: v \in e} \lambda(e,v) = \omega(v) \forall v \in V\}$$
This set is referred to as \( W_G \).
\end{definition}

Next, we will study how the conservation of liquidity in payment channels reduces the volume of feasible wealth distributions \( W_G \) for a given payment channel network.

\begin{lemma}
\label{lem:subset}
Let \( G(V,E,cap) \) be a payment channel network. For \( n > 2 \), there exists \( \omega \in W(C,n) \) such that \( \omega \notin W_G \).
\end{lemma}

\begin{proof}
We prove this by contradiction. Let \( e = (u,v) \in E \) be an arbitrary channel with capacity \( c_e > 0 \). Since \( n > 2 \), there is a user \( w \) such that \( u \neq w \neq v \). We define \( \omega: V \longrightarrow \mathbb{N} \) as:

\begin{equation}
\omega(x) = \begin{cases}
C & \text{if } x = w \\
0 & \text{otherwise}
\end{cases}
\end{equation}

By construction, we have \( C = \sum_{x \in V} \omega(x) \). Thus, \( \omega \) is feasible in \( W(C,n) \), and specifically, \( \omega(u) = \omega(v) = 0 \).

To conclude the proof, assume \( \omega \in W_G \). Due to the conservation of liquidity (equation \ref{eq:conservationOfLiquidity}) and the definition of \( \omega \) as a feasible wealth function in \( G \), we know that \( \omega(u) + \omega(v) \geq c_e \). This contradicts the assumption that the capacity of the channel was non-zero.
\end{proof}

\begin{corollary}
The number \( |W_G| \) of feasible wealth distributions of \( C \) coins among \( n > 2 \) users in a payment channel network is strictly smaller than the number \( |\mathcal{W}(C,n)| \) of wealth distributions that are feasible with on-chain transactions.
\end{corollary}

We are interested in how the topology of \( G \) impacts the relative volume \( r(G) \) that \( W_G \) has in comparison to \( \mathcal{W}(C,n) \):

\begin{equation}
r(G) = \frac{|W_G|}{|\mathcal{W}(C,n)|}
\label{eq:relVolume}
\end{equation}

For example, consider the network from the previous section in Figure \ref{fig:statePolytopeExample}. We can observe the feasible region \( W_G \) as a subpolytope of \( \mathcal{W}(21,3) \) in Figure \ref{fig:wealthPolytope}.

\begin{figure}[h]
\centering
\includegraphics[width=0.5\textwidth]{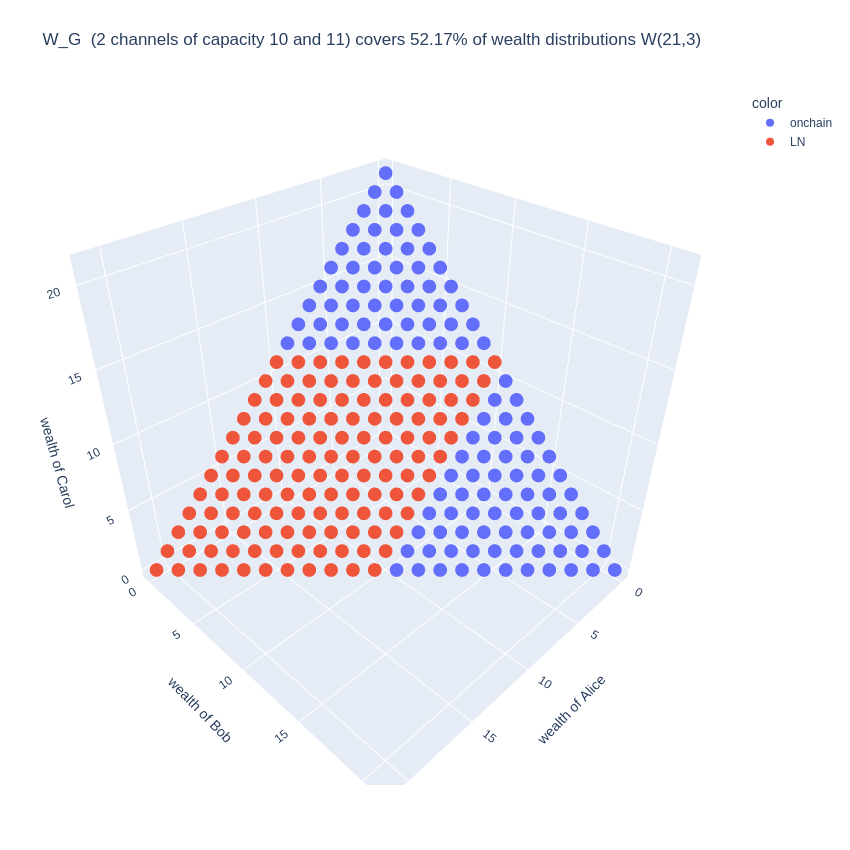}
\caption{Only 52.17\% of all wealth distributions in \( \mathcal{W}(21,3) \) are feasible if Alice, Bob, and Carol allocate the 21 coins into 2 channels of capacities 10 and 11.}
\label{fig:wealthPolytope}
\end{figure}

We Note that the feasible reagion changes with the topology of the network.
If Bob Spliced out 7 coins of his channel with Alice and created a new channel of 7 coins the feasible reagion would increase in size as can be seen in figure \ref{fig:feasibleLNWealthPolytopes}(a).
However the increase is relatively small.
It becomes larger if all channels had the same capcity as can be seen in figure \ref{fig:feasibleLNWealthPolytopes} (b).

\begin{figure}[h]
  \centering
  \subfigure[]{\includegraphics[width=0.23\textwidth]{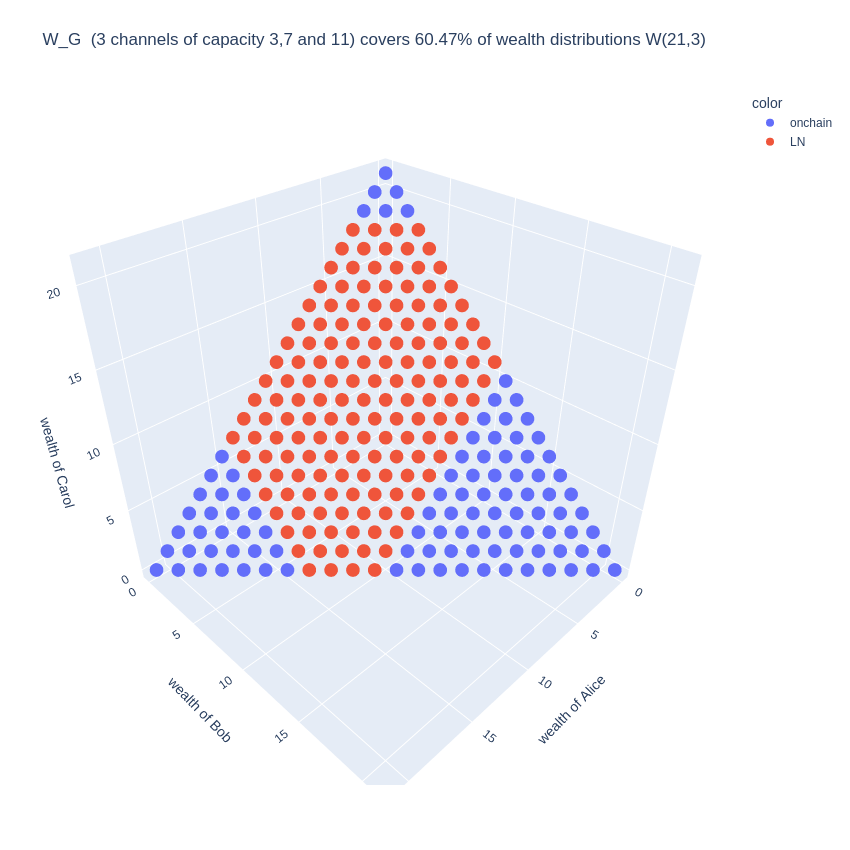}}
  \subfigure[]{\includegraphics[width=0.23\textwidth]{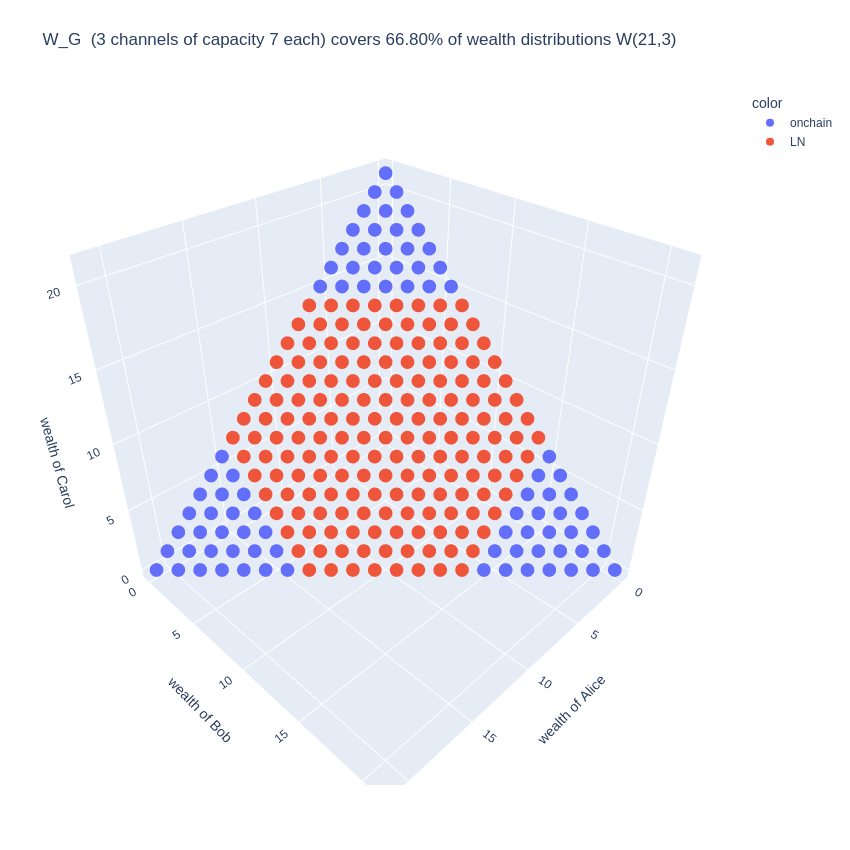}}
  \caption{$2$ feasible regions of $W_G$ for $2$ different networks between Alice Bob and Carol who own in total $21$ coins.}
  \label{fig:feasibleLNWealthPolytopes}
\end{figure}

We are not aware of an analytically closed formula that describes the volume \( |W_G| \) for a given network \( G \).
In particular we are not aware of a smart way to find the topology $G$ that maximizes $r(G)$.
However, we can estimate the volume through statistical sampling via Monte Carlo methods. Utilizing the Dirichlet Rescale Algorithm \cite{griffin2020generating}, we can uniformly and randomly select wealth distributions \( \omega \) from \( \mathcal{W}(C,n) \).\footnote{An open source implementation can be found at: \url{https://github.com/dgdguk/drs}} We can then check how often these distributions are feasible for a given network \( G \). This gives us an estimate for \( r(G) \):
\begin{equation}
r(G) \approx \frac{\text{number of \( \omega \) such that \( \omega \) is feasible in \( G \)}}{\text{number of sampled \( \omega \in \mathcal{W}(C,n) \)}}
\end{equation}
In order to be able to evaluate this we need to be able to decide if a given wealth distribution in $\mathcal{W}(C,n)$ is feasible in $W_G$.

\subsection{Deciding if a Wealth Distribution $\omega\in\mathcal{W}(C,n)$ is also feasible in $W_G$}
\label{sec:isomorphism}

Due to Lemma \ref{lem:subset}, we know there exists at least one vector \( \omega \in \mathcal{W}(C,n) \) that is infeasible and not in \( W_G \).  
For an arbitrary vector \( \omega \in \mathcal{W}(C,n) \), we can determine if it is also an element of \( W_G \) and thus a feasible wealth distribution on the payment channel network.  
According to the definition of a feasible wealth distribution in \( W_G \), we need to find \( \lambda \in L_G \) such that for all \( e \in E \) and \( v \in V \), the following \( n \) linear equations hold:  
\begin{equation}
  \label{eq:wealthConstraints}
  \sum_{e \in E: v \in e}\lambda(e,v) = \omega(v)
\end{equation}
Due to the conservation of liquidity, \( \lambda \) requires an additional \( m \) equations for all \( e \in E \):  
\begin{equation}
  \label{eq:channelConstraints}
  \lambda(e,v) + \lambda(e,u) = \text{cap}(e)
\end{equation}
Thus, we need to solve a system of \( n + m \) linear equations in \( 2 \cdot m \) variables:  
\[ \{e_{1,{u_1}}, e_{1,v_1}, \dots, e_{m,u_m}, e_{m,v_m} \} \]  
These equations are not fully independent because summing the \( n \) equations from Equation \ref{eq:wealthConstraints} yields the same value as summing the \( m \) equations from Equation \ref{eq:channelConstraints}:  
\[ \sum_{v \in V}\sum_{e \in E: v \in e}\lambda(e,v) = C = \sum_{e=(u,v) \in E}\left(\lambda(e,v) + \lambda(e,u)\right) \]  
Therefore, we can eliminate one of the constraints, resulting in a total of \( m + n - 1 \) independent equations.  
Let \( \sigma(G,\omega) \) be the solution space of feasible liquidity functions on \( G \) for the above system of linear equations over integers.  
If solved over \( \mathbb{Z}^{2 \cdot m} \), the dimension of the solution space \( \sigma(G,\omega) \) is at least \( 2 \cdot m - (m + n - 1) = (m - n + 1) \). 
However, we are seeking solutions over \( L_G \), which is isomorphic to the bounded hypercube:  
\[ H_G = \{0, \dots, c_1\} \times \{0, \dots, c_m\} \]  
Thus, for a wealth distribution to be feasible and a liquidity function \( \lambda \) to exist, we must require:  
\begin{equation}
  \sigma(G,\omega) \cap H_G \neq \emptyset
\end{equation}
Finding a solution for the system of linear equations over integers in a bounded region can be done through integer linear programming.  
\footnote{As discussed with Stefan Richter. Instead of solving a system of linear equations over a bounded region, one could solve a max flow problem to test feasibility. For this, one would take an arbitrary wealth vector \( \omega' \) and compute \( \omega - \omega' \). The components of the difference are the supply and demand of the nodes in the network. If a multi-source, multi-sink flow exists that fulfills the supply and demand, then \( \omega \in W_G \).}

\subsubsection{A Remark on Credit in Payment Channel Networks}
\label{sec:credit}
The previous geometric example illustrates that the system of linear equations that come from conservation of liquidity and the network topology always has at least one solution over $\mathbb{Z}^{2\cdot m}$.
However if the intersection of the solution space \(\sigma(G,\omega) \cap H_G\) is empty this means that all elements of $\sigma(G,\omega)$ have at least once component that is negative. 
This shows that some wealth distributions are infeasible on payment channel networks because the liquidity of the peers in the channel has to be always positive.
If credit was allowed in channels a feasible network state could be constructed for any wealth distribution.
A similar result has been shown before the Lightning Network was created \cite{dandekar2011liquidity}
Of course credit requirs trust and is consequently not wished for by the users and developers of peer to peer electronic payment systems.
However this geometric insight shows how difficult it is to create trustless electronic payment systems.

\subsection{Application: Deciding the Feasibility of Payments}

Making a payment is equivalent to changing the wealth vector.  
Let us assume \( w = (w_1, \dots, w_n) \in W_G \) is the current wealth distribution of the network.  
If the \( i \)-th user wishes to pay user \( j \) the amount \( a \), we can test the feasibility of the payment by checking if  
\begin{equation}
  \label{eq:payment}
  w' = w + a \cdot b_j - a \cdot b_i
\end{equation}
is still inside the polytope \( W_G \) of feasible wealth distributions.  
If \( w' \notin W_G \), then given our initial wealth distribution \( w \), the payment between \( i \) and \( j \) of amount \( a \) will fail.  

Remarkably, deciding whether a payment is feasible in a given network without conducting additional on-chain operations does not depend on the exact state of the network.  
Instead, the feasibility of a payment depends only on the network topology and the current wealth distribution.  
Of course, similar to the liquidity state of the network, the current wealth distribution is unknown. 

Thus we take a global perspective.
Similarly to estimating \( r(G) \), we can estimate the likelihood that a payment is feasible through Monte Carlo methods.
For this, we uniformly randomly sample several feasible wealth distributions in \( W_G \) and compute how often the resultant distribution after executing the payment is feasible.\footnote{This method has already been shared publicly before this paper was published.\url{https://delvingbitcoin.org/t/estimating-likelihood-for-lightning-payments-to-be-in-feasible/973}}
Node operators can use this to decide where to allocate liquidity an which channels to open or close.  

\begin{definition}
  We call $\rho$ the expected rate of infeasible payments.
\end{definition}

In particular, due to Lemma \ref{lem:subset}, it is noted that in payment channel networks, it is impossible for 100\% of all conceivable payment requests to be feasible.  

\begin{figure}[h]
\centering
 \includegraphics[width=0.40\textwidth]{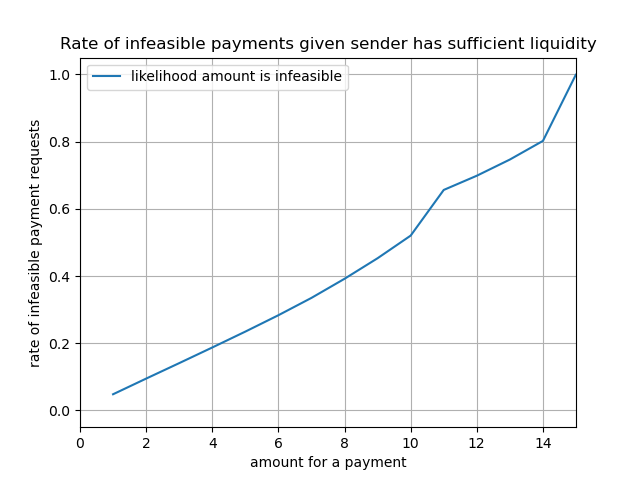}

\caption{Expected Rate \(\rho\) of infeasible payments in the example network between Alice, Bob, and Carol with channels of sizes 3, 7, and 11}
\label{fig:infeasible}
\end{figure}

\subsection{Limited Ability of Payment Channel Networks to scale Blockchains}

We emphasize that infeasible payment requests require at least one on-chain transaction to be executed successfully.  
This transaction can either change the wealth distribution and state of the payment channel network via an on-chain/off-chain swapping service or alter the network topology by opening and closing channels or through splicing.  
If \( \zeta \) is the number of possible on-chain transactions per second, we can derive the maximum bandwidth of supported Lightning Network payments as:  
\begin{equation}
  \label{eq:feasabilityRate}
\text{supported payments per second} = \mathcal{S} = \frac{\zeta}{\rho}
\end{equation}

In the Lightning Network whitepaper \cite{poon2016bitcoin}, the authors state that the Visa Network supports 47,000 payments per second during peak times.  
Solving the equation for \( \rho \) and noting that the Bitcoin network allows \( \zeta = 7 \) transactions per second, we conclude that only \( \rho = \frac{7}{47,000} \approx 0.0149\% \) of all conceivable payments can be infeasible if the Lightning Network is supposed have a comparable bandwidth to the visa network.  

Thus, for payment channel networks to scale the payment throughput of blockchains, it is crucial that the rate \( \rho \) of infeasible payments is close to zero.  
This happens when the feasible region of wealth distributions \( W_G \) is large within \( W(C,n) \) or if \( r(G) \) is close to $1$.  

\begin{figure}[h]
\centering
 \includegraphics[width=0.41\textwidth]{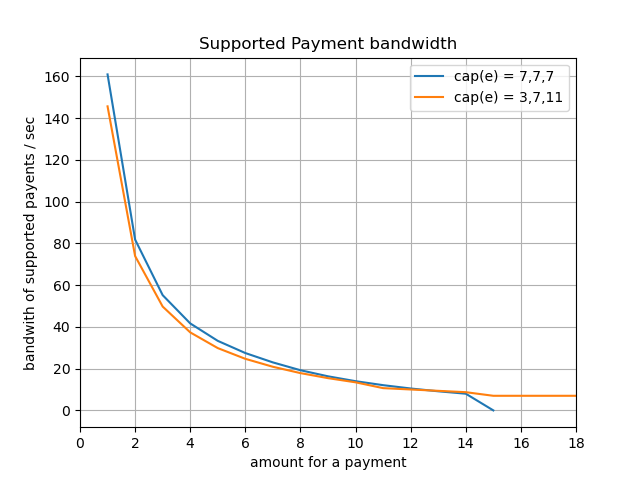}
\caption{Rate of infeasible payments \( \rho \) depending on the desired payment amount \( a \) in the example network between Alice, Bob, and Carol with channels of sizes 3, 7, and 11}
\label{fig:bandwidth}
\end{figure}

Using the data from Figure \ref{fig:infeasible} we can derive $\mathcal{S}$ for various networks and payment amounts which is depicted in figure \ref{fig:bandwidth}
The figure also demonstrates what users of payment channel networks have experienced before:
Such networks are mainly capable to scale the bandwidth of small amount payments.
There seems a small tradeoff between higher bandwidth but smaller supported maximum sendable amounts for networks with equally sized channels.
We will see later in this document how multiparty channel networks values of \( r(G) \) that are close to $1$.

\section{Lowering Expected Rate of Infeasible Payments via Multi Party Channels}
\label{sec:mpc-cuts}
We have seen in formular \ref{eq:feasabilityRate} that the supported payments per second in off-channel networks is proportional to the supported number of on-chain transactions multiplied with the inverse rate of expected infeasible payments. In particular we have seen annectotely that for two party channel networks the expected rate of infeasible payments is rather large. In the following we present an argument of how increasing the party size will allow more feasible wealth distributions and thus lower the expected rate of infeasible payments.% which is the only way to increase the supported bandwidth of feasible payments in an off-chain network without changing the on chain throughput. 

In order to do so we formalize the capital-efficiency advantage of $k$-party channels (sometimes refered to as coinpools or multiparty channels) over $2$-party channels.
The key observation is that for any node subset $S\subset V$, the maximal total wealth of $S$ must lie in an interval whose \emph{width equals the cut capacity} across $S$.
We start with the special case of $2$-party channels. Then we replace those $2$-party channels by $k$-party channels, which strictly increases these widths and thereby enlarges $W_G$. Obviously if the polytope of feasible wealth distributions increases more payments become feasible even if the on chain throughput is not changed.

\subsection{Cut intervals: A Proxy for the Number of Feasible Wealth Distributions}
\label{subsec:cut-intervals}

For $S\subset V$, let $E[S]:=\{\{u,v\}\in E:\ u,v\in S\}$ be internal edges and
$\delta(S):=\{\{u,v\}\in E:\ u\in S,\ v\notin S\}$ be crossing edges. We define the cut capacity
$C(\delta(S))$ as $C:=\sum_{e\in\delta(S)} c_e$.

\begin{lemma}
\label{lem:cut-interval}
For every $\lambda\in L_G$ with $\omega(v) = \sum_{e\in E,v\in e}\lambda(e,v)$ and every nonempty proper $S\subset V$,
\begin{equation}
\label{eq:cut-interval}
\sum_{e\in E[S]} c_e \;\;\le\;\; \sum_{v\in S}\omega(v) \;\;\le\;\; \sum_{e\in E[S]} c_e \;+\; \sum_{e\in \delta(S)} c_e.
\end{equation}
\end{lemma}

\begin{proof}
Compute
\begin{equation}
  \begin{split}
\sum_{v\in S}\omega(v) & =\sum_{v\in S}\sum_{v\in e}\lambda(e,v) \\
& =\underbrace{\sum_{e=(u,v)\in E[S]}\big(\lambda(e,u)+\lambda(e,v)\big)}_{=\sum_{e\in E[S]}c_e} \\
& +
\;\sum_{e\in\delta(S),u\in e\cap S}\lambda\big(e,u)
\end{split}
  \end{equation}
The last sum is between $0$ and $\sum_{e\in\delta(S)}c_e$, giving \eqref{eq:cut-interval}. Varying $\lambda(e,\cdot)$ on
$e\in\delta(S)$ sweeps the full interval while keeping $E[S]$ fixed.
\end{proof}

\begin{example}
  Let us look at the special case of $S$ being the set containing a single but arbitrary node $v\in V$ so $S=\{v\}$. In that case $E[S] = E[\{v\}]  = \emptyset$ as there are no internal edges for a single node. On the other hand the channels belonging to $v$ are the crossing edges $\delta(S) = \delta(\{v\})$. Applying Lemma \ref{lem:cut-interval} we see directly that the following holds for the total wealth $\omega(v)$ of $v$.
  \begin{equation}
    0 = \sum_{e\in E[S]}c_e \leq \omega(v) \leq \sum_{e\in\delta(\{v\})}c_e.
    \end{equation}
  The later is usually known as the total capacity of $v$ but in this case is also the same as the width of the cut interval of $S$. Thus the lemma describes in a more general setting what we already know and understand intuitively about the number of sats a node $v$ can own in given fixed topology.
  \end{example}

Lemma~\ref{lem:cut-interval} shows: \emph{every cut contributes a single scalar inequality} that constrains $\omega(S):=\sum_{v\in S}\omega(v)$ to a closed interval whose width is exactly $C(\delta(S))$.
%This can be seen easily when substracting $\sum_{e\in E[S]} c_e$ from equation \ref{eq:cut-interval} as we get:
%\begin{equation}
%\label{eq:cut_interval_2}
%0 \;\;\le\;\; \underbrace{\sum_{v\in S}\omega(v)}_{=w(S)} \;\;\le\;\; \sum_{e\i%n \delta(S)} c_e =: C(\delta(S)).
%\end{equation}

\subsection{Expected cut width for random topologies in $2$-party channel networks.}
We need to study the cut width probabilistically for a more realistic picture.
Let $n=|V|$ and fix $S$ with $|S|=s$. We can pick a two party channel uniformly from from $\binom{n}{2}$ pairs of nodes and study how likely it is that the picked channel will contribute to the cut interval.

\begin{lemma}
\label{lem:2pc-exact}
A random undirected two party channel $e=(u,v)$ crosses $(S,\bar S)$ with probability $\frac{2s(n-s)}{n(n-1)}$. Hence the expected cut interval for a $2$-party chanenl entwork has a width of
\[
\EE[C(\delta(S))] \;=\; m\,c\cdot \frac{2s(n-s)}{n(n-1)}.
\]
In particular, for a single node with $|S| = s=1$,
\(
\EE[C(\delta(S))] = m\,c\cdot \frac{2}{n}.
\)
\end{lemma}

\begin{proof}
  For a random channel $e$ to cross the $(S,\bar S)$ cut one node needs to be in $S$ and the other one in $\bar S$.
  In total here are $\binom{n}{2}$ pairs which could build a channel out of those $s\cdot(n-s)$ unordered pairs fulfill the above criterium.
  Thus $P[e\in\delta(S)] = \frac{2s(n-s)}{n(n-1)}$

  To compute the expected cut capacity for $m$ channels we draw channels without replacement.
  Let $X_i=1$ if the $i$-th channel crosses the cut and otherwise $0$.
  The number of crossings is $X=\sum_{i=1}^m{X_i}$.
  The expected number of crossings is:
  \begin{equation}
    \EE[X] = \sum_{i=1}^m\EE[X_i] = \sum_{i=1}^mPr[X_i = 1] = m\cdot P[e_i\in\delta(S)] 
  \end{equation}
  Plugging in the probability and multiplying the expected number of edges in the cut with the fixed capacity $c$ we get the expected cut capacity of a two party channel network to be:
  \begin{equation}
m\cdot c\frac{2s(n-s)}{n(n-1)}
    \end{equation}
  which had to be shown. In the special case with $s=1$ we get:
  \begin{equation}
m\cdot c\frac{2(n-1)}{2n(n-1)} = m\cdot c\frac{2}{n}
    \end{equation}
  which concludes the proof.
\end{proof}

We note that for scalability reasons with respect to onchain bandwidth we cannot support densly connected graph but need a sparse network with $m\approx n$.
As we have seen the cut capacity is an upper bound for the intervall width of $\omega(S)$.
Thus the expected range for $\omega(v)$ for a peer $v$ in a two party channel network is twice the ratio of channels $m$ to peers $n$ multiplied with the average capacity of a channel.
If a modification of the network (like introducing multiparty channels) increases $C(\delta(S))$ for some $S$, it widens that interval and cannot shrink the volume of the feasible wealth set $W_G$.
Increasing many cuts simultaneously relaxes many inequalities at once, strictly enlarging $W_G$ generically.

\subsection{Expected Cut Capacity in Multiparty Channel Networks}
\label{subsec:hyperedge-cut}

We will now generalize these results and show that the expected cut capacity for $k$-party channels ist strictly larger than for two party channels. More generally the expected cut capacity increases monotonically with $k$.
This indicates that $k+1$-party channels are preferable to $k$-party channels from a liquidity perspective. 

We model a $k$-party channel (also sometimes refered to as a coinpool) $e^k = \{v_{i_1},\dots,v_{i_k}\}\subset V$ of size $k$ and total capacity $c_{e^k}$ as an hyperedge.
As with two party channels if $e^k$ intersects both $S$ and $\bar S$, then the multi party channel can shift all $c_{e^k}$ across the cut by internal reassignment.
Thus it contributes $c_{e^k}$ to the cut width of the hypergraph.
If $c_{e^k}\subset S$ or $c_{e^k}\subset\bar S$, it contributes $0$ to the width and $c_{e^k}$
to the corresponding internal sum in~\eqref{eq:cut-interval}.

For the following we randomly fix $m$ arbitrary $k$-party channels, each of capacity $c$.\footnote{The choice of a uniform hypergraph seems like a strong assumption. It was chosen to show the mathematical principle with rather simple formulars at hand. Yet without proof we believe a different topology could be choosen to carry out a similar argument.}

%or a $k$-pool uniformly from $\binom{n}{k}$ $k$-sets
%and we compare the special case of the two-party channel network ($k=2$) with that of an arbitrary channel size of $k$ as a $k$-uniform hypergraph with $m$ hyperedges.

\begin{lemma}
\label{lem:kpc-exact}
A random $k$-set straddles $(S,\bar S)$ with probability
\[
q_k(s) \;:=\; 1 - \frac{\binom{s}{k}+\binom{n-s}{k}}{\binom{n}{k}}.
\]
Hence
\[
\EE[C(\delta(S))] \;=\; m\,c\, q_k(s).
\]
\end{lemma}

\begin{proof}
  There are $\binom{n}{k}$ potential random hyperedges in total. 
  A hyperedge fails to cross the cut if and only if all its vertices lie in $S$ or all lie in $\bar S$.
  The probability that all vertices lie in $S$ is $\frac{\binom{s}{k}}{\binom{k}{n}}$. The probability that all vertices lie in $\bar S$ is $\frac{\binom{n-s}{k}}{\binom{n}{k}}$
  Substracting the sum of those probabilities from 1 is exactly the probability that a multi party channel straddles to the $(S,\bar S)$ cut.

  For the expecte value we use the same argument as in lemma \ref{lem:2pc-exact}
which concludes the proof.
\end{proof}

We emphasize that $q_2(s)$ mathches the formerly derrived probability $\frac{2s(n-s)}{n(n-1)}$ for two party channels. To see this plug $k=2$ into
\[
q_k(s)=1-\frac{\binom{s}{k}+\binom{n-s}{k}}{\binom{n}{k}}
\]
and simplify:

\begin{equation}
\begin{aligned}
q_2(s)
&=1-\frac{\binom{s}{2}+\binom{n-s}{2}}{\binom{n}{2}} \\
&=1-\frac{\frac{s(s-1)}{2}+\frac{(n-s)(n-s-1)}{2}}{\frac{n(n-1)}{2}}\\
&=1-\frac{s(s-1)+(n-s)(n-s-1)}{n(n-1)}\\
&=1-\frac{s^2-s+,(n^2-2ns+s^2)-n+s}{n(n-1)}\\
&=1-\frac{(2s^2-2ns)+n^2-n}{n(n-1)}\\
&=\frac{n(n-1)-(2s^2-2ns+n^2-n)}{n(n-1)}\\
&=\frac{2ns-2s^2}{n(n-1)}
=\boxed{\frac{2s(n-s)}{n(n-1)}}.
\end{aligned}
\end{equation}

\subsubsection{Monotonicity with increasing $k$}

in order to show the desired monotonicity we need a lemma first.
\begin{lemma}\label{lem:binom-step}
Let integers satisfy $0\le k < s \le n$. Then
\[
\frac{\binom{s}{k+1}}{\binom{n}{k+1}}
\;=\;
\frac{\binom{s}{k}}{\binom{n}{k}}\cdot \frac{s-k}{\,n-k\,}.
\]
Equivalently,
\[
\frac{\binom{s}{k+1}}{\binom{s}{k}}
\;=\;
\frac{s-k}{k+1}
\qquad\text{and}\qquad
\frac{\binom{n}{k}}{\binom{n}{k+1}}
\;=\;
\frac{k+1}{\,n-k\,}.
\]
\end{lemma}

\begin{proof}
By definition $\displaystyle \binom{t}{r}=\frac{t!}{r!\,(t-r)!}$ for integers $t\ge r\ge 0$.

First,
\[
\frac{\binom{s}{k+1}}{\binom{s}{k}}
=\frac{\frac{s!}{(k+1)!(s-k-1)!}}{\frac{s!}{k!(s-k)!}}
=\frac{k!(s-k)!}{(k+1)!(s-k-1)!}
=\frac{s-k}{k+1}.
\]
furthermore,
\[
\frac{\binom{n}{k}}{\binom{n}{k+1}}
=\frac{\frac{n!}{k!(n-k)!}}{\frac{n!}{(k+1)!(n-k-1)!}}
=\frac{(k+1)!(n-k-1)!}{k!(n-k)!}
=\frac{k+1}{\,n-k\,}.
\]
Multiply these two equalities gives:
\[
\frac{\binom{s}{k+1}}{\binom{s}{k}}\cdot\frac{\binom{n}{k}}{\binom{n}{k+1}} =\frac{s-k}{k+1}\cdot \frac{k+1}{\,n-k\,} \\
\Leftrightarrow
\frac{\binom{s}{k+1}}{\binom{n}{k+1}}
\;=\;
\frac{\binom{s}{k}}{\binom{n}{k}}\cdot \frac{s-k}{\,n-k\,}.
\]
which is the desired identity.
\end{proof}

\begin{corollary}\label{cor:binom-step-mirror}
For $0\le k < n-s \le n$,
\[
\frac{\binom{n-s}{k+1}}{\binom{n}{k+1}}
\;=\;
\frac{\binom{n-s}{k}}{\binom{n}{k}}\cdot \frac{(n-s)-k}{\,n-k\,}.
\]
\end{corollary}

\begin{theorem}
\label{thm:qk-monotone}
For all $1\le s\le n-1$ and $2\le k\le n-2$,
\(
q_{k+1}(s)\ge q_k(s).
\)
Equivalently, the expected cut width is nondecreasing in $k$ for \emph{every} fixed cut.
\end{theorem}

\begin{proof}
To show $q_{k+1}(s)\ge q_k(s)$ it is equivalent to show
\[
\frac{\binom{s}{k+1}+\binom{n-s}{k+1}}{\binom{n}{k+1}}
\;\le\;
\frac{\binom{s}{k}+\binom{n-s}{k}}{\binom{n}{k}}.
\]
Using lemma \ref{lem:binom-step} and corollary \ref{cor:binom-step-mirror} we can write:

\begin{equation}
  \begin{aligned}
1-q_{k+1}(s) &= \frac{\binom{s}{k+1}+\binom{n-s}{k+1}}{\binom{n}{k+1}}\\
&= \frac{\binom{s}{k+1}}{\binom{n}{k+1}} +\frac{\binom{n-s}{k+1}}{\binom{n}{k+1}}\\
&= \frac{\binom{s}{k}}{\binom{n}{k}}\cdot\underbrace{\frac{s-k}{n-k}}_{\leq 1} +\frac{\binom{n-s}{k}}{\binom{n}{k}}\cdot\underbrace{\frac{(n-s)-k}{n-k}}_{\leq 1}\\
&\leq \frac{\binom{s}{k}}{\binom{n}{k}} +\frac{\binom{n-s}{k}}{\binom{n}{k}}\\
& = 1- q_{k}(s)
  \end{aligned}
  \end{equation}

Multiplying $1-q_{k+1}(s) \leq 1-q_{k}(s)$ with $-1$ and adding $1$ we get the desired $q_k(s) \leq q_{k+1}(s)$ which concludes the proof.
\end{proof}

The monotonicity of the probability that an hyperedge straddles to the cut set transfers directly to the expected width of the expected cut capacity.
As with $2$-party channels we study what this means for the expected cut capacity of single node by setting $S = \{v\}$ to be a set of size $s=1$ with a single arbitrary node $v$. 

Recall
\[
q_k(s)=1-\frac{\binom{s}{k}+\binom{n-s}{k}}{\binom{n}{k}}.
\]

For $s=1$ and $k\ge 2$ we have
 $\binom{1}{k}=0$. Thus for $k\ge 2$,
  \[
  q_k(1)=1-\frac{\binom{n-1}{k}}{\binom{n}{k}}
  =1-\frac{n-k}{n}
  =\boxed{\frac{k}{n}}.
  \]

Summary:
\[
\boxed{q_k(1)=\begin{cases}
0,&k=1 \footnote{For the edge case of $k=1$ we have: $q_1(1)=1-\frac{\binom{1}{1}+\binom{n-1}{1}}{\binom{n}{1}}=1-\frac{1+(n-1)}{n}=0$. This makes sense as $1$-party channel means there are not channels which means the cutwidth should be $0$ for any cut.}
\\
k/n,&2\le k\le n
\end{cases}}
\]
So for $s=1$, the cross cut straddle probability is \textbf{linear} in $k$ (starting at $0$ for $k=1$, then $(2/n,3/n,\dots,n/n = 1)$.

This means that the maximum expected wealth of a node $v$ in a $k$-party channel network with $m$ channels of capacity $c$ is:

\[
m\cdot c \frac{k}{n}
\]

This is consistent with the previously derived formular for $2$-party channels. Staying with $k$ party channels and assuming again a sparse network with $m\approx n$ we can see that the expected maximum wealth a node can have in a $k$-party channel network is proportional to $c\cdot k$.

In particular for fixed $n$ and $m$ we can see that the expected cut interval for every nodes grows linearly with $k$.

This result aligns nicely with the extrem case of $m=1$ and $k=n$. In that case all peers are members of a single channel with capacity $c$. In that case every node can have access to all the capacity $c$. Thus all wealth distributions are feasible which means every payment is feasible.

\subsection{Impact on feasible payment flows via max flow/min cut theorem}

By formula \eqref{eq:payment}, a payment of amount $a$ from $s$ to $t$ is feasible iff the post-payment
wealth vector remains feasible. Equivalently in the liquidity graph, the payment is feasible iff the $s\!\to t$ maxflow is at least $a$.
By the max–flow/min–cut theorem \cite{elias1956note,FordFulkerson1956},
\[
\operatorname{maxflow}(s,t)\;=\;\min_{S\subset V:\ s\in S,\ t\notin S}\ C(\delta(S)),
\]

Thus any transformation that (weakly) increases every $s$–$t$ cut capacity cannot decrease
$\operatorname{maxflow}(s,t)$, so at least as many payments are feasible.
For entire wealth targets $\omega$, the same cut capacities appear on the right endpoints of the
cut–interval constraints
\[
\sum_{e\in E[S]} c_e \;\le\; \omega(S) \;\le\; \sum_{e\in E[S]} c_e \;+\; C(\delta(S)),
\]
so enlarging cut widths can only enlarge the feasible wealth set $W_G$.

\subsubsection{Coupling of $k$-party channels to two-party channels.}
For any event $\mathcal{E}$ we write its indicator as
\[
\mathbf{1}\{\mathcal{E}\}
=\begin{cases}
1,&\text{if $\mathcal{E}$ holds},\\
0,&\text{otherwise}.
\end{cases}
\]

Fix a node set $V$ with $|V|=n$ and an index set of channels $e=1,\dots,m$.
Each channel $e$ carries a (possibly heterogeneous) capacity $c_e>0$.

We construct, on a common probability space, a $k$-party model and a paired two-party model as follows:
for each channel $e$,
\begin{enumerate}
\item sample a $k$-subset $P_e\subset V$ uniformly from $\binom{n}{k}$ (interpreted as the endpoint set of a $k$-party channel)
\item conditional on $P_e$, sample one unordered pair $E_e\in\binom{P_e}{2}$ uniformly (interpreted as the endpoints of a two-party channel).
\end{enumerate}
For a cut $(S,\bar S)$ with $S\subset V$, define the event that the $k$-party channel straddles the cut as:
\[
A_e(S):=\big(P_e\cap S\neq\emptyset\ \text{and}\ P_e\cap\bar S\neq\emptyset\big)
\]
and the event that the corresponding two-party channel crosses the cut
\[
B_e(S):=\big(E_e\ \text{has one endpoint in $S$ and one in $\bar S$}\big)
\]

\begin{lemma}\label{lem:indicator-dominance}
For every cut $S\subset V$ and every channel index $e$,
\[
\mathbf{1}\{A_e(S)\}\ \ge\ \mathbf{1}\{B_e(S)\}.
\]
\end{lemma}

\begin{proof}
If the two-party channel $E_e$ crosses the cut, then its two endpoints lie on opposite sides; since $E_e\subset P_e$, the $k$-party channel $P_e$ necessarily contains vertices on both sides, i.e., $A_e(S)$ holds. Thus $B_e(S)\Rightarrow A_e(S)$, which implies the indicator inequality.
\end{proof}

\begin{corollary}\label{cor:cut-capacity-dominance}
Let the cut capacity contributed by channel $e$ to $(S,\bar S)$ be
\[
C^{(k)}_e(S):=c_e\,\mathbf{1}\{A_e(S)\},\qquad
C^{(2)}_e(S):=c_e\,\mathbf{1}\{B_e(S)\}.
\]
Then for every cut $S$,
\[
\sum_{e=1}^m C^{(k)}_e(S)\ \ge\ \sum_{e=1}^m C^{(2)}_e(S).
\]
In particular, the minimum cut satisfies
\[
\min_{S\subset V}\ \sum_{e=1}^m C^{(k)}_e(S)\ \ge\ \min_{S\subset V}\ \sum_{e=1}^m C^{(2)}_e(S).
\]
\end{corollary}

\begin{proof}
Sum the inequality in Lemma~\ref{lem:indicator-dominance} over $e$ and weight by $c_e$.
Taking the minimum over all cuts preserves the inequality.
\end{proof}

\begin{corollary}\label{cor:maxflow-feasibility}
For all ordered pairs $(s,t)\in V\times V$, $s\neq t$,
\[
\operatorname{maxflow}_{k\text{-party}}(s,t)\ \ge\ \operatorname{maxflow}_{\text{2PC}}(s,t)\footnote{more generally we could have conducted the same argument not from $k$-party channels to $2$-party but from $(k+1)$-party channels to $k$-party channels. All the set theoretic arguments stay the same in the general case.},
\]
by the max–flow/min–cut theorem. Moreover, the feasible-wealth set defined by the cut–interval constraints
\[
\sum_{e\in E[S]} c_e \;\le\; \omega(S) \;\le\; \sum_{e\in E[S]} c_e \;+\; \sum_{r=1}^m C^{(\cdot)}_r(S)\qquad(\forall S\subset V)
\]
is monotone in the cut widths, hence
\[
W_G^{(k\text{-party})}\ \supseteq\ W_G^{(\text{2PC})}.
\]
\end{corollary}

\subsection{Implementing Multiparty Channels}
It seems as if there are three main burdens of multi party channels.
\begin{enumerate}
\item coordination overhead and interactivity requirements
\item unilateral exit costs - in particular on chain foot print
  \item membership problem: How can a node easily join a multi-party channel?
\end{enumerate}
At the time of writing we are aware of two reasonable approaches that come in question when looking at multi party channels. First, we have eltoo \cite{decker2018eltoo} which would require a softfork of the bitcoin protocol. Secondly, Ark\footnote{\url{https://bitcoinops.org/en/topics/ark/}} style systems which are actively being worked on have been proposed as they seem to work with a rather low interactivity requirements\cite{Pickhardt2025} and could thus work as a channel factory\cite{burchert2018scalable}. 

\section{Relation between the Polytope $L_G$ of Liquidity States and $W_G$ of feasible Wealth Distributions}
Given the topology of the lightning network as a weighted, undirected graph $G(V,E,cap)$.
There is a geometric relation between the corresponding polytopes $L_G$ of liquidity states and $W_G$ of feasible wealth distributions.
For every feasable state $\lambda \in L_G$ we have seen that the $n$ participants $\{v_1,\dots,v_n\}$ have a non negative wealth stored in their channels.
Thus we can project any liquidity state to a feasible wealth distribution via the projection:

\begin{equation}
  \begin{split}
  \label{eq:pi}
  \pi: L_G & \longrightarrow {W_G}\\
  \lambda & \stackrel{\pi}{\longmapsto}\sum_{e=(u,v)\in E}\left(\lambda(e,u)\cdot b_u + \lambda(e,v)\cdot b_v\right)
  \end{split}
\end{equation}

We proof that the image of $\pi$ is indeed a subset of $W_G$. To do this we show that $\pi(\lambda)$ is feasible in $W_G$.

\begin{lemma}
Let $w = \pi(\lambda)$. For a given base we can write $w=\sum_{v\in V}w_v\cdot b_v$. Then $\sum_{v\in V}w_v = C$ 
\end{lemma}
\begin{proof}
  We use the definition of our projection $\pi$ in equation \ref{eq:pi} as well as conservation of liquidity (equation \ref{eq:conservationOfLiquidity}) and the definition of $C$:

\begin{equation}
\begin{split}
  \sum_{v\in V}w_v & = \sum_{v\in V}\sum_{e\in E:v \in e}\lambda(e,v) \\
  & = \sum_{e\in E}\sum_{v\in e}\lambda(e,v) \\
  & = \sum_{e\in E} c_{e} \\
  & = C
\end{split}
\end{equation}
\end{proof}

%% Because of the methods introduced in section \ref{sec:isomorphism} we can find a preimage of $w\in W_G$ under $\pi$.
%% Thus $\pi$ is surjective. 
%% In topology any surjective map induces an equivalence relation from which a quotiont space can be constructed.

%% \begin{definition}
%% We call two states $\lambda,\mu\in L_G$ equvalent if and only they are projected to the same wealth distribution id est: $\pi(\lambda)=\pi(\mu)$.
%% In this case we write $\lambda\sim_{\pi}\mu$.
%% \end{definition}

%% \begin{lemma}
%% The relation $\sim_{\pi}$ is an equivalance relation. In particular it is reflexive, symmetrical and transitve.
%% \end{lemma}
%% \begin{proof}
%%   We need to show that the three properties for an equivalence relation are fulfilled:
%%   \begin{enumerate}
%%   \item \textbf{Reflexivity:} Because $\pi(\lambda)=\pi(\lambda)$ we have $\lambda\sim_{\pi}\lambda$ for any $\lambda \in L_G$.
%%   \item \textbf{Symmetry:} Let $\lambda,\mu\in L_G$ with $\lambda\sim_{\pi}\mu$ which means that $\pi(\lambda)=\pi(\mu)$. Because equality is symmetical it follwos that $\pi(\mu)=\pi(\lambda)$ which means that $\mu\sim_{\pi}\lambda$
%%   \item \textbf{Transitivity:} Let $\lambda,\mu,\nu\in L_G$ with $\lambda\sim_{\pi}\mu$ and $\mu\sim_{\pi}\nu$. We have $\pi(\lambda) = \pi(\mu) = \pi(\nu)$ from which it follows that $\lambda\sim_{\pi}\nu$.
%%   \end{enumerate}
  
%% \end{proof}

Due to the methods introduced in Section \ref{sec:isomorphism}, we can find a preimage of \( w \in W_G \) under \( \pi \).  
Thus, \( \pi \) is surjective.  
In topology, any surjective map induces an equivalence relation, from which a quotient space can be constructed.  

\begin{definition}
We call two states \( \lambda, \mu \in L_G \) equivalent if and only if they are projected to the same wealth distribution, i.e., \( \pi(\lambda) = \pi(\mu) \).  
In this case, we write \( \lambda \sim_{\pi} \mu \).
\end{definition}

\begin{lemma}
The relation \( \sim_{\pi} \) is an equivalence relation. In particular, it is reflexive, symmetrical, and transitive.
\end{lemma}

\begin{proof}
We need to show that the three properties of an equivalence relation are fulfilled:
\begin{enumerate}
    \item \textbf{Reflexivity:} Since \( \pi(\lambda) = \pi(\lambda) \), we have \( \lambda \sim_{\pi} \lambda \) for any \( \lambda \in L_G \).
    \item \textbf{Symmetry:} Let \( \lambda, \mu \in L_G \) with \( \lambda \sim_{\pi} \mu \), which means \( \pi(\lambda) = \pi(\mu) \). Since equality is symmetrical, it follows that \( \pi(\mu) = \pi(\lambda) \), which means \( \mu \sim_{\pi} \lambda \).
    \item \textbf{Transitivity:} Let \( \lambda, \mu, \nu \in L_G \) with \( \lambda \sim_{\pi} \mu \) and \( \mu \sim_{\pi} \nu \). We have \( \pi(\lambda) = \pi(\mu) = \pi(\nu) \), from which it follows that \( \lambda \sim_{\pi} \nu \).
\end{enumerate}
\end{proof}

%% \begin{definition}
%% We call $[\lambda] = \{\mu \in L_G | \lambda\sim_{\pi} \mu\}$ the equivalence class of $\lambda$.
%% The quotient space $L_G/\sim_{\pi}$ is the space of equivalence classes.
%% \end{definition}

%% In particular the equivalence class $[\lambda]=\pi^{-1}(\{w\})$ is the preimage of $w\in W_G$ under the projection $\pi$. 
%% It follows that
%% \begin{equation}
%%   L_G/\sim_{\pi}\cong W_G
%% \end{equation}
%% for which we have explicitly provided the isomorphism.
%% In one direction it is just the projection of the channel states to the wealth distribution as defined in $\pi$ and in the other direction it is constructing a feasible state for a feasible wealth distribution as described in section \ref{sec:isomorphism}.

\begin{definition}
We call \( [\lambda] = \{\mu \in L_G \mid \lambda \sim_{\pi} \mu\} \) the equivalence class of \( \lambda \).  
The quotient space \( L_G / \sim_{\pi} \) is the space of equivalence classes.
\end{definition}

In particular, the equivalence class \( [\lambda] = \pi^{-1}(\{w\}) \) is the preimage of \( w \in W_G \) under the projection \( \pi \).  
It follows that
\begin{equation}
  L_G / \sim_{\pi} \cong W_G
\end{equation}
for which we have explicitly provided the isomorphism.  
In one direction, it is the projection of the channel states to the wealth distribution as defined by \( \pi \), and in the other direction, it is the construction of a feasible state for a feasible wealth distribution as described in Section \ref{sec:isomorphism}.

\subsection{Rebalancing Payment Channels and Circulations}
\label{sec:rebalancing}

We aim to study the number of elements in the equivalence class \( |[\lambda]| \).  
Recall that for \( \lambda, \mu \in [\lambda] \), we have \( \pi(\lambda) = \pi(\mu) \).  
Thus, \( \lambda \) and \( \mu \) project to the same wealth distribution.  
When routing fees are ignored, circular self-payments, also known as channel rebalancings, are the only payments that do not change the wealth distribution.  
However, they do alter the liquidity state of the network.  

We will now prove that the size of the equivalence class \( |[\lambda]| \) is equal to the number of strict circulations that exist on \( \mathcal{L}(G, \lambda) \).  
Also, if \( \pi(\lambda) = w \), then \( |\pi^{-1}(\{w\})| \) is equal to the number of circulations.\footnote{We are not the first to note that circular self-payments do not change the wealth distributions of the participants in a payment channel network \cite{piatkivskyi2018rebalancing}.}  
Counting the number of circulations is possible with Ehrhart's theory \cite{ehrhart1962polyhedra} or by applying Barvinok's algorithm \cite{barvinok1994polynomial}.  
In particular, free open source software\footnote{\url{https://www.math.ucdavis.edu/~latte/}} exists \cite{de2004effective} that can achieve this.  

To follow these observations, we review a few elements from the theory of network flows.  

\begin{definition}
A flow network is a directed graph \( G = (V, E) \) with a capacity function \( c: E \rightarrow \mathbb{Z}^+ \) and a flow function \( f: E \rightarrow \mathbb{Z} \) such that \( 0 \leq f(e) \leq c(e) \) for all \( e \in E \).
\end{definition}

Certain flows are of particular interest:  

\begin{definition}
A flow \( f \) is called a circulation if at every vertex \( x \in V \) we have:
\begin{equation}
\sum_{(u,x) \in E} f(u,x) = \sum_{(x,v) \in E} f(x,v)
\end{equation}
\end{definition}

%% Because of the particular way how we created the liquidity graph $\mathcal{L}(G,\lambda)$ we need to get a stricter definition for circulations.

%% \begin{definition}
%%   On bidirectional flow Networks where for any $(u,v)\in E$ there is also $(v,u)\in E$ we call a circulation $f$ strict if for all edges $(u,v)\in E$ we have either $f(u,v)=0$ or $f(v,u)=0$ or both $f(u,v)=0=f(v,u)$.
%% \end{definition}

%% We now show that there is a $1$ to $1$ correspondence between strict circulations on $\mathcal{L}(G,\lambda)$ and the elements of $[\lambda]$.
%% First, let $f$ be a circulation on $\mathcal{L}(G,\lambda)$.
%% We show that $\lambda'$ that is defined pointwise for every $e=(u,v)$ through
%%   \begin{equation}
%%   \lambda'(e,u) = \lambda(e,u) + f(v,u) - f(u,v)
%%   \end{equation}
%%   and
%%   \begin{equation}
%%   \lambda'(e,v) = \lambda(e,v) + f(u,v) - f(v,u)
%%   \end{equation}
%%   is indeed a member of $[\lambda]$.
%%   For this we proof the following lemma.

%% \begin{lemma}
%%   Let $G(V,E,cap)$ be a payment channel network in an arbitrary feasible state $\lambda\in L_G$.
%%   Then $\lambda'\in L_G$ and in particular $\pi(\lambda)=\pi(\lambda')$.
%% %  If the flow function $f$ on the liquidity network $\mathcal{L}(G,\lambda)$ is a circulation then $\pi(\lambda') = \pi(\lambda)$ for 
%% \end{lemma}
%% \begin{proof}
%%   We use the fact tha $\lambda \in L_G$ and conservation of liquidity to show that conservation of liquidity also holds for $\lambda'$: 
Due to the particular construction of the liquidity graph \(\mathcal{L}(G,\lambda)\), we require a more precise definition of circulations.

\begin{definition}
  In bidirectional flow networks, where for any \((u,v) \in E\) there is also \((v,u) \in E\), a circulation \(f\) is called strict if for all edges \((u,v) \in E\), we have either \(f(u,v) = 0\) or \(f(v,u) = 0\), or both \(f(u,v) = 0 = f(v,u)\).
\end{definition}

We now demonstrate that there is a one-to-one correspondence between strict circulations on \(\mathcal{L}(G,\lambda)\) and the elements of \([\lambda]\).  
First, let \(f\) be a circulation on \(\mathcal{L}(G,\lambda)\).  
We show that \(\lambda'\), defined pointwise for every \(e = (u,v)\) by
\begin{equation}
  \lambda'(e,u) = \lambda(e,u) + f(v,u) - f(u,v)
\end{equation}
and
\begin{equation}
  \lambda'(e,v) = \lambda(e,v) + f(u,v) - f(v,u),
\end{equation}
is indeed a member of \([\lambda]\).  
To establish this, we prove the following lemma.

\begin{lemma}
  Let \(G(V,E,cap)\) be a payment channel network in an arbitrary feasible state \(\lambda \in L_G\).  
  Then \(\lambda' \in L_G\) and, in particular, \(\pi(\lambda) = \pi(\lambda')\).
\end{lemma}
\begin{proof}
  We use the fact that \(\lambda \in L_G\) and conservation of liquidity to show that conservation of liquidity also holds for \(\lambda'\):
%%%%% HERE CHAT GPT TRIED TO PROOF THE THEOREM. thile the computation is correct it doesn't proof the lemma. It can be used at a later spot though
  %% \begin{equation}
  %%   \sum_{e \in E: v \in e} \lambda'(e,v) = \sum_{e \in E: v \in e} \left(\lambda(e,v) + f(u,v) - f(v,u)\right)
  %% \end{equation}
  %% By definition of a strict circulation, this simplifies to:
  %% \begin{equation}
  %%   \sum_{e \in E: v \in e} \lambda(e,v) + \sum_{e \in E: v \in e} f(u,v) - \sum_{e \in E: v \in e} f(v,u)
  %% \end{equation}
  %% Given that \(f\) is a circulation, the net flow is zero, and we get:
  %% \begin{equation}
  %%   \sum_{e \in E: v \in e} \lambda(e,v) = \sum_{e \in E: v \in e} \lambda'(e,v)
  %% \end{equation}
  %% Thus, conservation of liquidity holds for \(\lambda'\) and therefore \(\lambda' \in L_G\).  
  %% Since \(\pi(\lambda') = \pi(\lambda)\), \(\lambda'\) belongs to the equivalence class \([\lambda]\).
%\end{proof}

\begin{equation*}
    \begin{split}
      \lambda'(e,u) + \lambda'(e,v) & = \lambda(e,u) + f(v,u) - f(u,v) \\
       & + \lambda(e,v) + f(u,v) - f(v,u)\\
      & = \lambda(e,u) + \lambda(e,v) \\
      & + \underbrace{f(u,v) - \underbrace{f(v,u) + f(v,u)}_{=0} - f(u,v)}_{=0}\\
      & = \lambda(e,u) + \lambda(e,v) \\
      & = cap(e)
    \end{split}
  \end{equation*}
  Furhter more we have: 
  \begin{equation}
    \lambda'(e,u) = \underbrace{\lambda(e,u) - \underbrace{f(u,v)}_{\leq \lambda(e,u)}}_{\geq 0} +\underbrace{f(v,u)}_{\geq 0} \geq 0
  \end{equation}
and
  \begin{equation}
    \lambda'(e,u) = \underbrace{\lambda(e,u) + \underbrace{f(v,u)}_{\leq cap(e) - \lambda(e,u)}}_{\leq cap(e)} -\underbrace{f(u,v)}_{\geq 0} \leq cap(e)
  \end{equation}
  This shows that $0\leq \lambda'(e,u) \leq cap(e)$ indicating $\lambda'\in L_G$
  
%%   To show that $\pi(\lambda)=\pi(\lambda')$ we recall the definition of $\lambda'$:
%%     $$\lambda' = \lambda + \underbrace{\sum_{(u,v)\in E'}\left(f(u,v)\cdot b_{(u,v),u} - f(u,v)\cdot b_{(u,v),v}\right)}_{\mu}$$
%%   Since $\pi$ is linear we have $\pi(\lambda') = \pi(\lambda) + \pi(\mu)$.
%%   Thus to show equality of $\pi(\lambda)=\pi(\lambda')$ it is sufficient to test $0\stackrel{!}{=}\pi(\mu)$

%%   \begin{equation}
%%     \begin{split}
%%     \pi(\mu) & =\sum_{(u,v)\in E}\left(f(u,v)\cdot b_u - f(u,v)\cdot b_v\right)
%%     \end{split}
%%   \end{equation}
%%   in particular the $x$-th component of $\pi(\mu)$ can be written as:
%%   $$(\pi(\mu))_x = \sum_{(x,v)\in E}f(x,v) - \sum_{(u,x)\in E}f(u,x)$$
%%   Because of conservation of flow and $f$ being a circulation this equation is $0$ and thus $(\pi(\mu))_x=0$.
%%   Therefor $\pi(\mu)=0$ and $\pi(\lambda)=\pi(\lambda')$
%% \end{proof}
To show that \(\pi(\lambda) = \pi(\lambda')\), we recall the definition of \(\lambda'\):
\[ \lambda' = \lambda + \underbrace{\sum_{(u,v) \in E} \left(f(u,v) \cdot b_{(u,v),u} - f(u,v) \cdot b_{(u,v),v}\right)}_{\mu} \]
Since \(\pi\) is linear, we have \(\pi(\lambda') = \pi(\lambda) + \pi(\mu)\). Thus, to show \(\pi(\lambda) = \pi(\lambda')\), it is sufficient to prove that \(\pi(\mu) = 0\).

\begin{equation}
    \begin{split}
    \pi(\mu) &= \sum_{(u,v) \in E} \left(f(u,v) \cdot b_u - f(u,v) \cdot b_v\right)
    \end{split}
\end{equation}

In particular, for an arbitrary $x\in V$ the \(x\)-th component of \(\pi(\mu)\) can be written as:
\[ (\pi(\mu))_x = \sum_{(x,v) \in E} f(x,v) - \sum_{(u,x) \in E} f(u,x) \]

Because of the conservation of flow and \(f\) being a circulation, this equation equals zero, so \((\pi(\mu))_x = 0\).

Therefore, \(\pi(\mu) = 0\) and \(\pi(\lambda) = \pi(\lambda')\).
\end{proof}

After establishing that a circulation \(f\) on \(\mathcal{L}(G,\lambda)\) results in a new liquidity state \(\lambda'\), we aim to show that the provided mapping is injective.

\begin{lemma}
  Let \(f \neq g\) be strict circulations on \(\mathcal{L}(G,\lambda)\). Then the associated liquidity states \(\lambda_f \neq \lambda_g\).
\end{lemma}

\begin{proof}
  We prove this by contradiction.
  Assume \(\lambda_f = \lambda_g\).
  This means that for all \(e = (u,v) \in E\), we have:
  \begin{equation*}
    \begin{split}
      \lambda(e,u) + f(v,u) - f(u,v) &= \lambda(e,u) + g(v,u) - g(u,v) \\
      \Leftrightarrow f(v,u) - f(u,v) &= g(v,u) - g(u,v)
    \end{split}
  \end{equation*}
  Because \(f\) and \(g\) are strict circulations, at least one term on each side of the equation equals zero. We consider all cases, noting that \(f,g \geq 0\):

  \textbf{1. Case:}
  \[ f(v,u) = g(v,u) = 0 \Rightarrow -f(u,v) = -g(u,v) \]

  \textbf{2. Case:}
  \[ f(u,v) = g(u,v) = 0 \Rightarrow f(v,u) = g(v,u) \]

  \textbf{3. Case:}
  \[ f(v,u) = g(u,v) = 0 \Rightarrow -f(u,v) = g(v,u) = 0 \]

  \textbf{4. Case:}
  \[ f(u,v) = g(v,u) = 0 \Rightarrow f(v,u) = -g(u,v) = 0 \]

  This shows that \(f = g\), which contradicts our assumption that \(f \neq g\). Therefore, \(\lambda_f \neq \lambda_g\).
\end{proof}

Finally, we prove that our correspondence between strict circulations and equivalent liquidity states is surjective.

\begin{lemma}
  For any \(\lambda' \in [\lambda]\), there exists a strict circulation \(f\) on \(\mathcal{L}(G,\lambda)\) such that \(\lambda'(u,v) = \lambda(u,v) + f(v,u) - f(u,v)\).
\end{lemma}
\begin{proof}
For any \(e = (x, y)\), we define:
\[
f(x, y) = 
\begin{cases}
    \lambda(e, x) - \lambda'(e, x) & \text{if } \lambda(e, x) \geq \lambda'(e, x) \\
    0 & \text{otherwise}
\end{cases}
\]
From the definition, it follows that:
\[
0 \leq f(x, y) \leq c(x, y)
\]

The \(x\)-th component of the wealth vector \(\pi(\lambda)\) is computed as:
\[
\left(\pi(\lambda)\right)_x = \sum_{e \in E: x \in e} \lambda(e, x)
\]
Since \(\lambda' \in [\lambda]\), we have \(\pi(\lambda') = \pi(\lambda)\). It follows for any \(x \in V\):
\[
0 = \sum_{e \in E: x \in e} \left(\lambda(e, x) - \lambda'(e, x)\right)
\]

If \(\lambda(e, x) - \lambda'(e, x) \geq 0\), then we have:
\[
f(x, y) = \lambda(e, x) - \lambda'(e, x)
\]

Otherwise, we use the conservation of liquidity to show:
\begin{equation*}
    \begin{split}
        \lambda(e, x) - \lambda'(e, x) &= c(e) - \lambda(e, y) - \left(c(e) - \lambda'(e, y)\right) \\
        &= c(e) - \lambda(e, y) - c(e) + \lambda'(e, y) \\
        &= -\lambda(e, y) + \lambda'(e, y) \\
        &= -\underbrace{\left(\lambda(e, y) - \lambda'(e, y)\right)}_{f(y, x)} \\
        &= -f(y, x)
    \end{split}
\end{equation*}

Therefore, we can replace \(\lambda(e, x) - \lambda'(e, x)\) with \(f(x, y) - f(y, x)\), yielding:
\begin{equation}
    \begin{split}
        0 &= \sum_{e \in E: x \in e} \left(\lambda(e, x) - \lambda'(e, x)\right) \\
        &= \sum_{e \in E: x \in e} \left(f(x, y) - f(y, x)\right) \\
        &= \sum_{e \in E: x \in e} f(x, y) - \sum_{e \in E: x \in e} f(y, x) \\
        &\Leftrightarrow \sum_{(y, x)} f(y, x) = \sum_{(x, y)} f(x, y)
    \end{split}
\end{equation}

This proves that \(f\) is a strict circulation on \(\mathcal{L}(G, \lambda)\).
\end{proof}

From these lemmas, the main theorem follows:

\begin{theorem}
  \label{thm:circulations}
    The fiber \(\pi^{-1}(\{w\})\) of any feasible wealth distribution \(w \in W_G\) is an equivalence class \([\lambda]\) of liquidity states. The number of strict circulations on \(\mathcal{L}(G, \lambda)\) is the same as the number \(|[\lambda]|\) of distinct liquidity states of the network with the same fixed wealth distribution \(w \in W_G\).
\end{theorem}

Circular rebalancing of liquidity does not change the feasibility of a payment since circulations leave the wealth distribution invariant, and the feasibility of payments is determined by testing if the change in wealth distribution still results in a feasible wealth distribution. However, depending on which liquidity state \(\lambda \in \pi^{-1}(\{w\})\) the network is in, the speed at which payment planning strategies of nodes can find the necessary liquidity for feasible payments can be impacted.

\begin{figure}[h]
  \centering
  \hspace{1.5em}
\subfigure[]{\includegraphics[width=0.152\textwidth]{polytope_of_wealth_distributions_lninequal}}
\hspace{1.5em}
\subfigure[]{\includegraphics[width=0.23\textwidth]{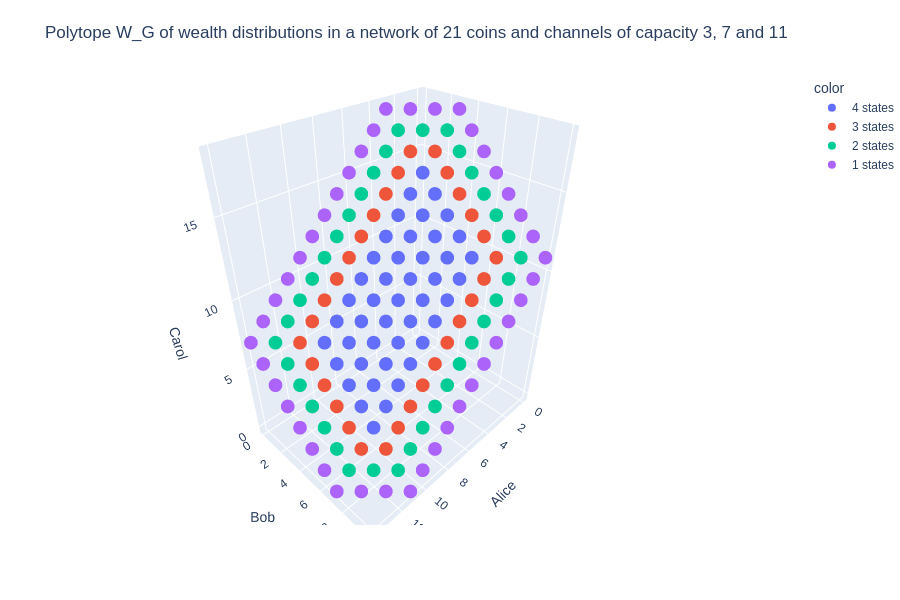}}
\subfigure[]{\includegraphics[width=0.23\textwidth]{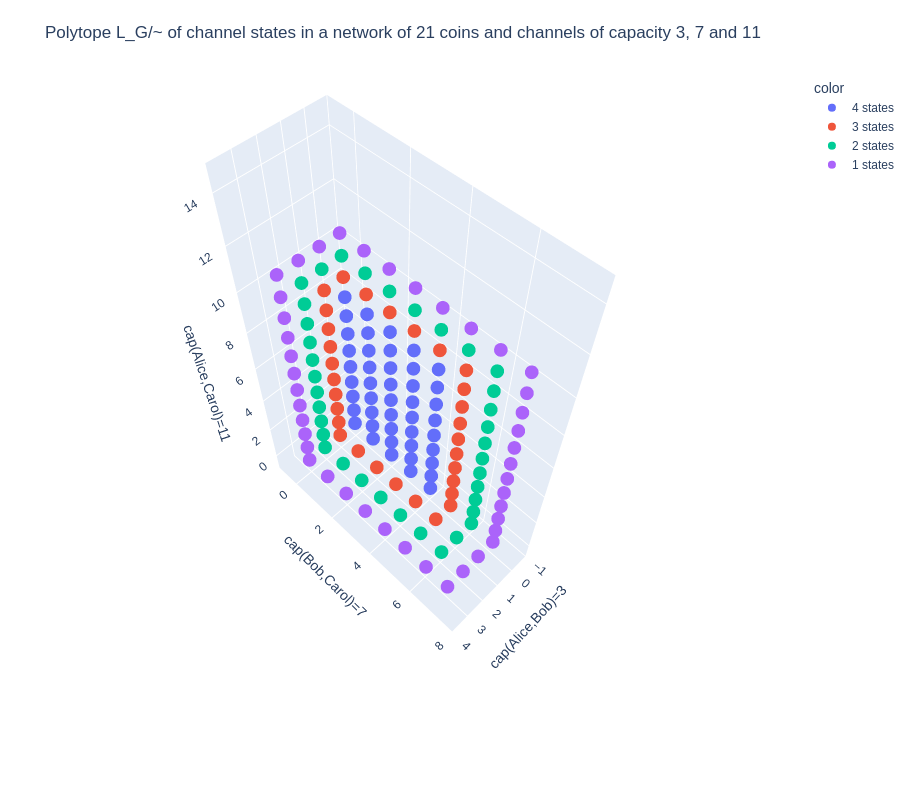}}
\subfigure[]{\includegraphics[width=0.23\textwidth]{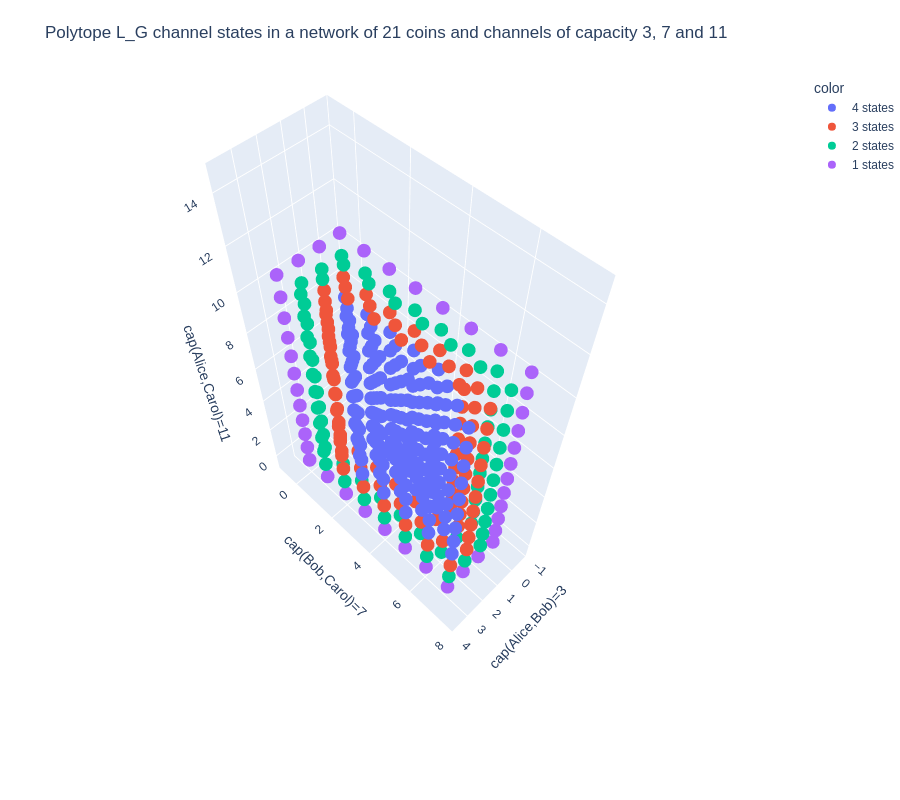}}
\caption{(a) shows the feasible region $W_G$ as in figure \ref{fig:feasibleLNWealthPolytopes} (b) depicts the size of the equivalence classe for each feasible wealth distribution. (c) Space $L_G/\sim_\pi$ of equivalance classes as a subset of the surface of $L_G$ (d) full polytope $L_G$ of all liquidity states}
\label{fig:equivalenceclasses}
\end{figure}

Figure \ref{fig:equivalenceclasses} demonstrates in $3$-dimensional geometry how the $L_G$ in combination with $\pi$ can be seen as a bundle over $W_G$ and how $W_G\cong L_G/\sim_\pi$.
In particular in subfigure (c) one can realize that the feasible region $W_G$ corresponds to just $3$ sides of the hypercube:
$$H_G=\{0,1,2,3\}\times\{0,2,3,4,5,6,7\}\times\{0,\dots,11\}$$
Each side corresponds to the subspace in which one dimension (state) equals $0$.
Meaning the liquidity state which is rebalanced so long until one channel is depleted. This obviously works only if cycles are present and not in trees.

\subsection{Special Case of Spanning Trees}
%% Let us assume $G$ had the shape of a tree.
%% In this case
%% $$m = |E| = |V|-1 = n-1$$
%% We recall the dimension of $\sigma(G,\omega)$ was: $m-n+1$.
%% replacing $m$ with $n-1$ we see that the dimension of $\sigma(G,\omega)$ is $m-n+1 = (n-1)+n+1 = n-n + 1-1 = 0$
%% A zero dimensional space is just a single point.
%% If $\omega$ was feasible this point lies in $\W_G$ and the liquidity function is unique.
%% Thus on tree shaped networks there is a one to one correspondance of the set $W_G$ and the set of feasible liquidity functions $L_G$. In particular we can write:
%% $$W_G \cong L_G$$
%% The careful reader may have realized that the polytope of feasible liquidity states $L_G$ in figure \ref{fig:statePolytopeExample} had the same number of points as the feasible reagion $W_G$ in figure \ref{fig:wealthPolytope}.
%% As we have seen this is no surprise but to be expected.

%% This means that for a uniform distribution of feasible wealth distributions in a spanning tree the liquidity in channels would also be uniformly distributed and depletion would not occure.
%% This is remarkable because it has been shown\cite{guidi2019steadystate} that under the assumption of balanced flows there is a stable state in which the liquidity within most channels ist depleted but there is a spanning tree in which the liquidity is uniformly distributed. 

Let us assume \(G\) has the shape of a tree. In this case,
\[ m = |E| = |V| - 1 = n - 1 \]
We recall that the dimension of \(\sigma(G, \omega)\) was: \(m - n + 1\). Replacing \(m\) with \(n - 1\), we see that the dimension of \(\sigma(G, \omega)\) is:
\[ m - n + 1 = (n - 1) - n + 1 = 0 \]
A zero-dimensional space is just a single point. If \(\omega\) is feasible, this point lies in \(W_G\), and the corresponding liquidity function is unique. Thus, in tree-shaped networks, there is a one-to-one correspondence between the set \(W_G\) and the set of feasible liquidity functions \(L_G\). In particular, we can write:
\[ W_G \cong L_G \]
The careful reader may have noticed that the polytope of feasible liquidity states \(L_G\) in Figure \ref{fig:statePolytopeExample} had the same number of points as the feasible region \(W_G\) in Figure \ref{fig:wealthPolytope}. As we have seen, this is not surprising but to be expected.

This means that for a uniform distribution of feasible wealth distributions in a spanning tree, the liquidity in channels would also be uniformly distributed, preventing depletion. This is remarkable because it has been shown \cite{guidi2019steadystate} that under the assumption of balanced flows, there is a stable state in which the liquidity within most channels is depleted, but there is a spanning tree in which the liquidity is uniformly distributed.

\section{Channel Depletion as a Consequence of rational economic behavior of nodes}
\label{sec:depletion}
So far we have neglected the existence of routing fees to study the possible liquidity states for a feasible wealth distribution given a channel network.
In practice node operators will charge a fee to fulfill forwarding requests.
\begin{definition}
  A function $f:E\times V\longrightarrow \mathbb{N}_0$ is called a fee function. It assignes each channel the ppm that the node operator charges to forward money on that channel. It neglects upfront fees and the base fee.
\end{definition}
At the same time all node implementations will include the minimization of routing fees that need to be paid to their objective cost function when selecting a payment min cost flow (or payment paths).
We see how node operators set their fees and try to allocate liquidity in their channels to maximize their fee potential.\footnote{\url{https://github.com/DerEwige/speedupln.com/blob/main/docs/fee_potential_and_rebalancing.md}}
For the following we neglect upfront fees and the basefee that node operators charge to fulfill a routing request.
\begin{definition}
  The \textbf{fee potential of a node} is defined through:
  \[
  p_v=\sum_{e\in E:v\in e}f(e,v)\cdot\lambda(e,v)
  \]
\end{definition}

Routing nodes try to set fees and move liquidity in a way that maximizes their fee potential.
\begin{definition}
The \textbf{fee potential of the network} is defined similiarly as:

\begin{equation}
p_G=\sum_{e\in E}\sum_{v\in e}f(e,v)\cdot\lambda(e,v)
\end{equation}
\end{definition}
In section \ref{sec:isomorphism} we have seen that that the solution of an integer linear program will help us to decide if wealth distribution is feasible on the network.
In section \ref{sec:rebalancing} we saw that there are often many liquidity states that fulfill the wealth distribution if cycles exist on the network.
Given rational economic behavior of nodes we assume that the most likely liquidity state of a feasible wealth distribution is one that maximizes the fee potential $p_G$ of the network.
The solution to this maximization problem can be found with the same linear program that is used to decide if a wealth distribution is feasible.

When generating random networks with random fee functions we can for a feasible wealth distribution (e.g. the center of the polytope of feasible wealth distributions) compute the liquidity state that maximizes the fee potential.
Doing so we observer that most channels will be depleted:

\begin{figure}[h]
\centering
 \includegraphics[width=0.41\textwidth]{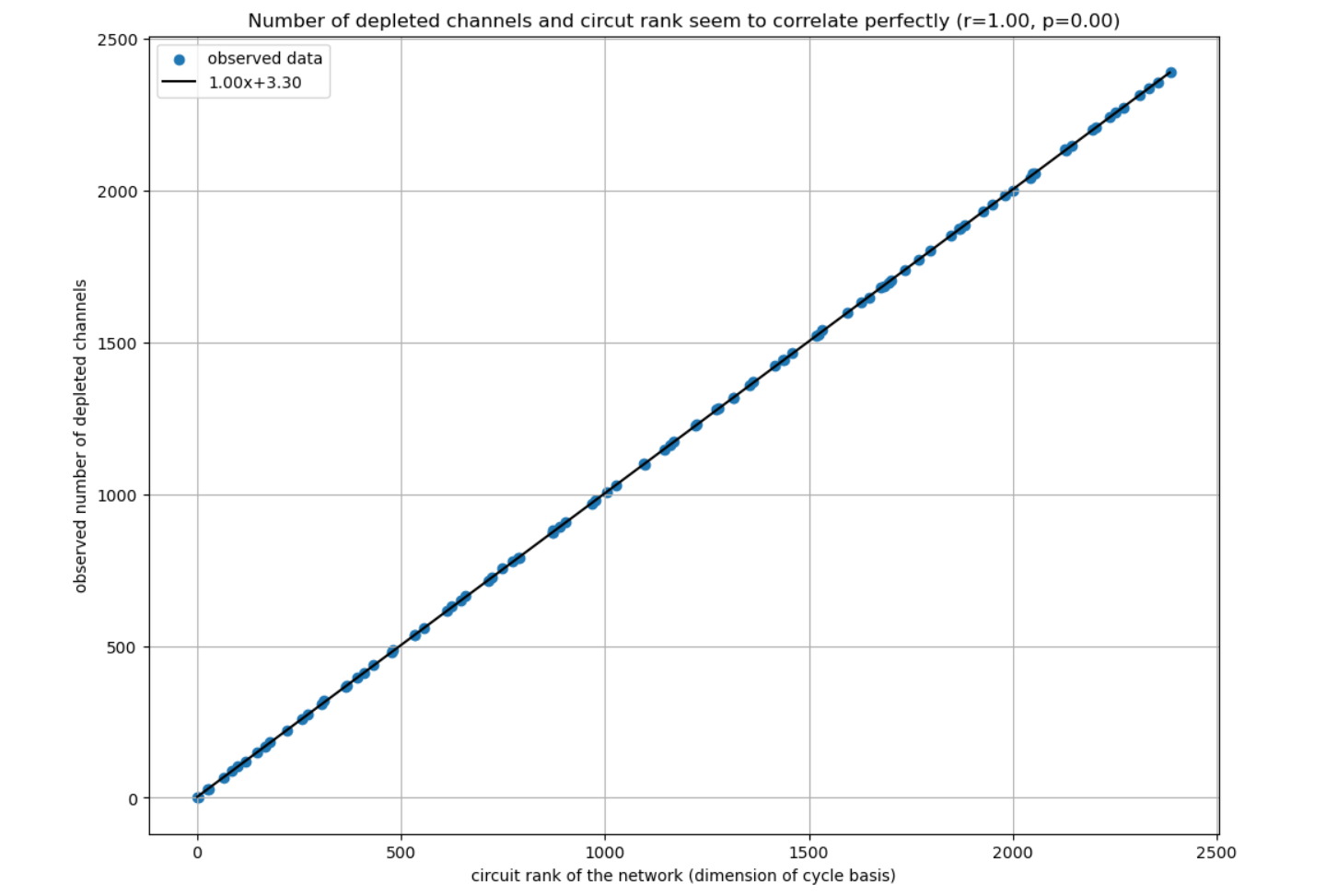}
\caption{The number of depleted channels correlates strongly with the cicruit rank of the network}
\label{fig:circuitrank}
\end{figure}

This diagram emperically underscores the results by \cite{guidi2019steadystate}.
In their research they assumed that the steady state of flows on channels across the network will let channels channels deplete besides a spanning tree.
In the following we will show how the same result can be derrived assuming the above mentioned linear program resulting from rational economic behavior of node operators.

\begin{lemma}\footnote{Before publication of this Lemma and proof Anastasios Sidiropoulos has shared with me the core argument of this lemma in a private conversation when I shared figure \ref{fig:circuitrank} in a mail discussion with him.}
  \label{lem:differenceOfFeesDepletion}
  Assume $\lambda\in L_G$ is a feasible liquidity state that satisfies a wealth distribution $\omega\in W_G$ and maximizes the fee potential $p_G$.
  Let $C=\{v_0,\dots,v_n=v_0\}$ be a non depleted cycle of channels on the network. This means $\forall i\in\{0,\dots,n-1\}$:
  \begin{enumerate}
  \item $(v_i,v_{i+1})\in E $.
  \item $\lambda((v_i,v_{i+1}),v_i)\neq 0 \neq \lambda((v_{i},v_{i+1}),v_{i+1})$
  \end{enumerate}
  Then:
  \[
  \sum_{i=0}^{n} ( f((v_i, v_{i+1}),v_i) - f((v_i, v_{i+1}),v_{i+1}) = 0
  \]
  This means the objective function does not change if money is rebalanced along the cycle.
\end{lemma}
\begin{proof}
  The contribution of the cycle to the objective function is given through:
    \begin{align*}
      p_C = & \sum_{i=0}^{n} f((v_i, v_{i+1}),v_i) \cdot \lambda((v_i,v_{i+1}),v_i)  \\
      + & \sum_{i=0}^{n} f((v_{i+1}, v_{i}),v_{i+1}) \cdot  \lambda((v_i,v_{i+1}),v_{i+1})
    \end{align*}
    Since all channels are non depleted we can send at least $x$ units with $|x|\geq 1$ along the cycle.
    This does not change the wealth distribution.
    However the new objective function computes to:
    \begin{align*}
      p_C' = & \sum_{i=0}^{n} f((v_i, v_{i+1}),v_i) \cdot (\lambda((v_i,v_{i+1}),v_i)-x)  \\
      + & \sum_{i=0}^{n} f((v_{i+1}, v_{i}),v_{i+1}) \cdot  (\lambda((v_i,v_{i+1}),v_{i+1})+x)
    \end{align*}
    Thus the difference $\delta = p_C - p_C'$ computes to:
    \begin{align*}
      \delta = & \sum_{i=0}^{n} f((v_i, v_{i+1}),v_i) \cdot x  \\
      + & \sum_{i=0}^{n} f((v_{i+1}, v_{i}),v_{i+1}) \cdot (-x)
    \end{align*}
    This is the same as:
    \begin{align*}
      \delta = x \cdot \underbrace{\sum_{i=0}^{n} f((v_i, v_{i+1}),v_i) - f((v_{i+1}, v_{i}),v_{i+1})}_{const}
    \end{align*}
    We note that $x$ can be positive or negative as none of the channels were depleted.
    Furthermore the sum of differences of fees is constant.
    Thus the change $\delta$ in the network's fee potential $p_G$ after rebalancing has to be $0$ otherwise $p_G$ was not maximal.
    Since $x$ was non zero we conclude that this is only possible if the constant sum of differences of fees is $0$ which concludes the proof
\end{proof}

The careful reader will note that there is an alternative condition such that the fee potential could be maximal but the sum of differences in fees along the circle is non $0$.
This happens exactly if the circle contains depleted channels and it is not possible to further rebalance in this direction.

\begin{corollary}
  Let the fee potential $p_G$ be maximal and $\delta = x \cdot \underbrace{\sum_{i=0}^{n} f((v_i, v_{i+1}),v_i) - f((v_{i+1}, v_{i}),v_{i+1})}_{const}$ be non zero for some $x\neq 0$.
  Then some channel along the cycle must be depleted e.g.:
  \begin{itemize}
  \item $\lambda((v_i,v_{i+1},v_i)) = 0$ or
  \item $\lambda((v_i,v_{i+1}),v_{i+1})=0$
  \end{itemize}
\end{corollary}

Of course in an arbitrary network the condition from the lemma and corollary are knife edge conditions and unlikely to take place.
Thus depletion of channels is expected to happen and the norm rather than the exception unless for a spanning forest whose exact location may vary with the current wealth distribution.

\section{Mitigation strategies for Channel Depletion}
As we have seen in the last section depletion of channels seems to be an inherent phenomenon given the design of the protocol and resulting incentives.
In this chapter we will utilize the previous results to describe $3$ different approaches to mitigate depletion by chaning the incentives through an upgrade of the protocol.
We sepecifically do not recommend any of those approaches as they yield consequences beyond the scope of this paper. We just present working solutions and explainations why they would work if adopted.

\subsection{Introducing symmetric fees on channels}
The technically easiest approach to mitigate depletion is to introduce and force symmetric routing fees for channels.
The previous section demonstrated that even within a circular economy and net 0 flow on all peers most channels are expected to deplete.
This is due to the fact that payments from a node $v$ to another node $u$ may take a different flow than the reverse payment from node $u$ to node $v$.
One very obvious way to mitigate this problem is by looking at lemma \ref{lem:differenceOfFeesDepletion}.
We recognized that depletion takes only place if the sum of fee differences along the channels in a cycle is not $0$.
Were the Lighting Network community to upgrade the protocol and force node operators to publish a single fee rate for a channel which would be charged for forwarding an HTLC in both directions the condition of the lemma would hold and there would not be any depletion presure emerging from the economically selfish behavior of sending nodes.
While technically easily to change the protocol in this way we assume that node operators would not be willing to accept such a change as they would have to agree with their channel peer on a fee that would be charged in both directions. 

\subsection{Convex (tiered) fees: discrete model and interior--boundary dynamics}
\label{subsec:convex-discrete}

Operators typically raise the forwarding fee on a channel $e=\{u,v\}$ as the local liquidity at an endpoint falls ($\lambda(e,u)\!\approx\!0$) and lower it as the channel fills ($\lambda(e,u)\!\approx\!c_e$). This scarcity pricing aims to throttle demand and avoid boundary hits. A natural abstraction is a tiered (discretely convex) schedule that becomes more expensive as remote liquidity increases (equivalently, as local liquidity decreases). In practice, such schedules are hard to deploy with source routing, because fees would need to change with state; fee rate cards~\cite{Neigut2022} emulate convexity with piecewise-linear tiers but do not solve the sender's lack of state visibility.

For each channel $e=(u,v)$ the fees are charged by the sending endpoint of a hop. Thus when using scarcity pricing the tier applied on hop $u\!\to\!v$ is a nonincreasing function $p_{e,u}$ of the sender’s local liquidity $\lambda(e,u)$\footnote{we could see it as increasing in the remote liquidity which is the total net liquidity that has been sent. In this sense the function can be seen as convex.}:
\[
p_{e,u}:\{0,1,\dots,c_e\}\to\mathbb{Z}_{\ge 0}
\]

Along a simple cycle $C=(v_0,\dots,v_\ell=v_0)$ with canonical orientation $v_i\!\to\!v_{i+1}$ and edges $e_i=(v_i,v_{i+1})$,
the strict circulation $\lambda\mapsto\lambda_x$ updates local endpoint liquidity by
\[
\lambda_x(e_i,v_i)=\lambda(e_i,v_i)-x
\]
and the remote liquidity of each channel by
\[
\lambda_x(e_i,v_{i+1})=\lambda(e_i,v_{i+1})+x
\]
with
\[x\in[x_{\min},x_{\max}]=\big[-\min_i \lambda(e_i,v_{i+1}),\ \min_i \lambda(e_i,v_i)\big]\cap\mathbb{Z}\]
Using the discrete potential
\[\Phi(\lambda)=\sum_{e=(u,v)}\sum_{t=1}^{\lambda(e,u)}p_{e,u}(t)+\sum_{t=1}^{\lambda(e,v)}p_{e,v}(t)
\]
the net one–unit gain in the cycle direction is
%\begin{equation}
  \begin{align*}
\Delta_C(x)\ & :=\ \Phi(\lambda_{x+1})-\Phi(\lambda_x) \\
& =\ \sum_{i=0}^{\ell-1}\Big[p_{e_i,v_{i+1}}\!\big(\lambda_x(e_i,v_{i+1})+1\big)\;-\;p_{e_i,v_i}\!\big(\lambda_x(e_i,v_i)\big)\Big]
\label{eq:deltaC-orient}
\end{align*}
%\end{equation}
This is next-unit price at each receiver (its local liquidity increased by $1$) minus refunded price at each sender. Mor importantly $\Delta_C(x)$ is monoton as we will see in the following lemma: 

\begin{lemma}
If every $p_{e,u}$ is nonincreasing in local liquidity $\lambda(e,u)$, then $\Delta_C(x)$ is nonincreasing in $x$ on $[x_{\min},x_{\max}-1]\cap\mathbb{Z}$.
\end{lemma}

\begin{proof}
Let $R_i(x)=\lambda_x(e_i,v_{i+1})$ and $S_i(x)=\lambda_x(e_i,v_i)$.
We define the per-edge contribution
\[
\Delta_i(x)\;:=\;p_{e_i,v_{i+1}}\big(R_i(x)+1\big)\;-\;p_{e_i,v_i}\big(S_i(x)\big)\]
so that
\[
\Delta_C(x)=\sum_{i=0}^{\ell-1} \Delta_i(x).
\]
When we increase the cycle push by one unit ($x\mapsto x+1$), the receiver’s local liquidity increases and the sender’s local liquidity decreases:
\[
R_i(x+1)=R_i(x)+1,\qquad S_i(x+1)=S_i(x)-1.
\]
Hence the one-step change of the $i$-th contribution is
\begin{align*}
& \Delta_i(x+1)-\Delta_i(x) \\
&=\Big[p_{e_i,v_{i+1}}\big(R_i(x+1)+1\big)-p_{e_i,v_{i+1}}\big(R_i(x)+1\big)\Big]\\
&\quad-\Big[p_{e_i,v_i}\big(S_i(x+1)\big)-p_{e_i,v_i}\big(S_i(x)\big)\Big]\\
  &=\underbrace{p_{e_i,v_{i+1}}\big(R_i(x)+2\big)-p_{e_i,v_{i+1}}\big(R_i(x)+1\big)}_{\text{receiver term} \leq 0}\\
& -
\underbrace{\big(p_{e_i,v_i}(S_i(x)-1)-p_{e_i,v_i}(S_i(x))\big)}_{\text{sender term}\geq 0}.
\end{align*}
Since we used scarcity pricing $p$ is nonincreasing in local liquidity. This means for any $x$ we have $p(x+1)\le p(x)$ and $p(x-1)\ge p(x)$.

Therefore
\[
\underbrace{p_{e_i,v_{i+1}}(R_i(x)+2)-p_{e_i,v_{i+1}}(R_i(x)+1)}_{\le 0}\ \le\ 0
\]
and
\[
\underbrace{p_{e_i,v_i}(S_i(x)-1)-p_{e_i,v_i}(S_i(x))}_{\ge 0}\ \ge\ 0,
\]
and thus
\[
\Delta_i(x+1)-\Delta_i(x)\ \le\ 0-0\ =\ 0.
\]
Summing over $i$ gives $\Delta_C(x+1)-\Delta_C(x)=\sum_i\big(\Delta_i(x+1)-\Delta_i(x)\big)\le 0$.
\end{proof}

Hence $\Delta_C$ changes its sign at most once across its feasible integers which can be used to proof the following:

\begin{lemma}
\label{lem:cycle-optimality}
Fix a simple cycle $C$ and the feasible integer interval $[x_{\min},x_{\max}]$.
If each tier $p_{e,u}$ is nonincreasing in local liquidity, then $\Delta_C(x)$ is nonincreasing in $x$.
Consequently:
\begin{enumerate}
\item (\emph{Interior balance}) If there exists $x^\star\in[x_{\min},x_{\max}-1]$ with
\[
\Delta_C(x^\star)\ \ge\ 0\ \ge\ \Delta_C(x^\star\!+\!1),
\]
then every minimizer of $x\mapsto \Phi(\lambda_x)$ on $[x_{\min},x_{\max}]$ lies in $\{x^\star,x^\star\!+\!1\}\subset (x_{\min},x_{\max})$.
\item (\emph{Boundary depletion}) If no such sign change exists, then the minimizer is at a boundary:
$\arg\min \Phi(\lambda_x)=\{x_{\min}\}$ when $\Delta_C>0$ everywhere, and
$\arg\min \Phi(\lambda_x)=\{x_{\max}\}$ when $\Delta_C<0$ everywhere.
\end{enumerate}
\end{lemma}

\begin{proof}
$\Delta_C$ nonincreasing means the discrete differences of $\Phi(\lambda_x)$ decrease with $x$, i.e.\ $\Phi(\lambda_x)$ is discrete convex on $[x_{\min},x_{\max}]$.
For any discrete convex function, the forward difference crossing through $0$ brackets all minimizers; if there is no crossing, the minimum occurs at an endpoint in the direction of the sign.
\end{proof}

We can interpret the sign change of $\Delta_C$ is a stable attactor: When it occurs, the cycle admits an interior operating equilibrium with no direction depleted. If there is no sign change that attractive equilibrium is outside the feasibility region and the channel will deplete to the boundary of the feasibility region i.e. some directed local liquidity hits zero. With strong enough convexity depletion should not happen. In the current case of a linear fee function the attraction point is at infinity and always outside the feasibility region. Thus depletion under linear fees is the norm.

Tiered (discretely convex) schedules can act as flow control: Opposing price tiers meet in the interior on many cycles, so opportunistic cycle pushes stall before any direction hits zero. However, utilization-responsive tiers conflict with source routing: The path cost depends on hidden, instantaneous $\lambda(e,\cdot)$ at every hop. Without state disclosure or probing, the sender cannot price paths correctly. Thus, while convex tiers can stabilize liquidity and mitigate depletion, broad deployment likely requires protocol support beyond static gossip of fee rates.

\subsubsection{Experiment: Depletion with linear vs.\ quadratic fees}
\label{subsubsec:exp-convex-vs-linear}

To demonstrate the effect we implemented a small synthetic network with fixed capacities, balanced initialization, and a net-zero (circular economy) demand model with unit-sized payments and selfish, source-based path selection. Fees were evaluated on disclosed schedules; attempts succeeded iff all liquidity along the chosen path were sufficient; successful payments updated endpoint liquidity; no explicit rebalancing was performed during a run.

Per endpoint $(e,v)$ we compared two different fee schedules.
\begin{itemize}
\item \textbf{Linear (today’s proportional fee intuition):}
Define a total fee on an edge as
\(F(x) = ppm\cdot x\)
Sending one additional sat on an a channel will always cost you \emph{ppm}.

\item \textbf{Quadratic (ppm as coefficient of a quadratic total-fee):}
Define a total fee on a channel $e=(u,v)$ edge as 
\(F(x) = ppm\cdot \frac{x^2}{c_e}\)
In this case the routing becomes more expensive for every sat that is routed on the channel. In particular the cost depends on $\lambda$. To see this we compute\begin{equation}
\begin{aligned}
\Delta F & = F(\lambda(e,v) + 1) - F(\lambda(e,v)) \\
 & = \frac{ppm}{c_e}\cdot(2\lambda(e,v)+1) \approx \frac{ppm}{c_e}2\cdot\lambda(e,v)
\end{aligned}
\end{equation}
In particular if $\lambda(e,v)=c_e/2$ we have $\Delta F = ppm$ which is the same unit cost as the linear fee schedule. If more liquidity is on $u$'s side of the channel the cost to rout a satoshi is cheaper than \emph{ppm}, otherwise it is  more expensive. 
This shows why we devide the feerate by the total capacity of the channel. We wanted to stay comparable with the \emph{ppm} values of today's proportinal \emph{ppm} interpretation.
\end{itemize}

\begin{figure}[h]
\centering
\includegraphics[width=0.5\textwidth]{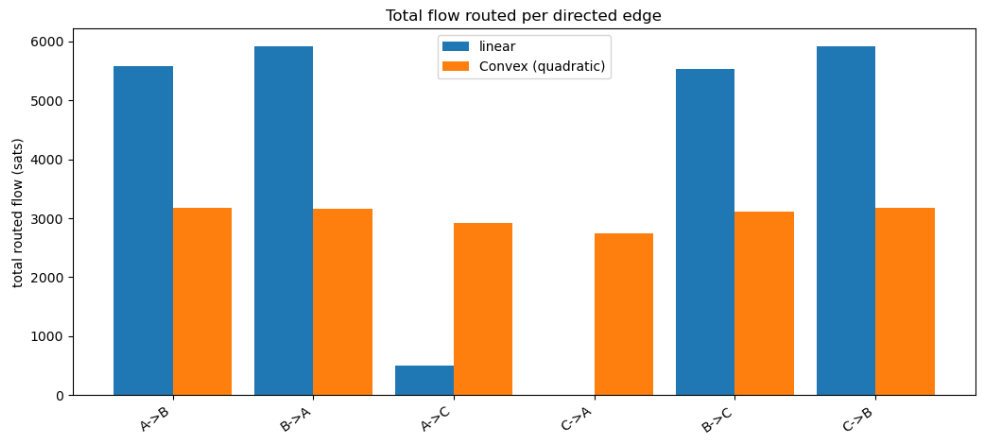}
\caption{We can see that with a convex fee schedule all channels attract a similar amount of flow while some channels do not attract flow in the linear fee setting.}
  \label{fig:convexFlowControl}
\end{figure}

We crated an artificial sample network consisting of 3 nodes that formed a single cycle to demonstrate the effect of a convex fee schedule. In figure \ref{fig:convexFlowControl} we can see nicely that all channel have roughly the same amount of flow if a convex fee schedule is used while some channels attract no flow if a linear fee schedule is being used. This emperically confirms that flow controll via perfectly convex fee schedules is possible. 

\begin{figure}[h]
\centering
\includegraphics[width=0.5\textwidth]{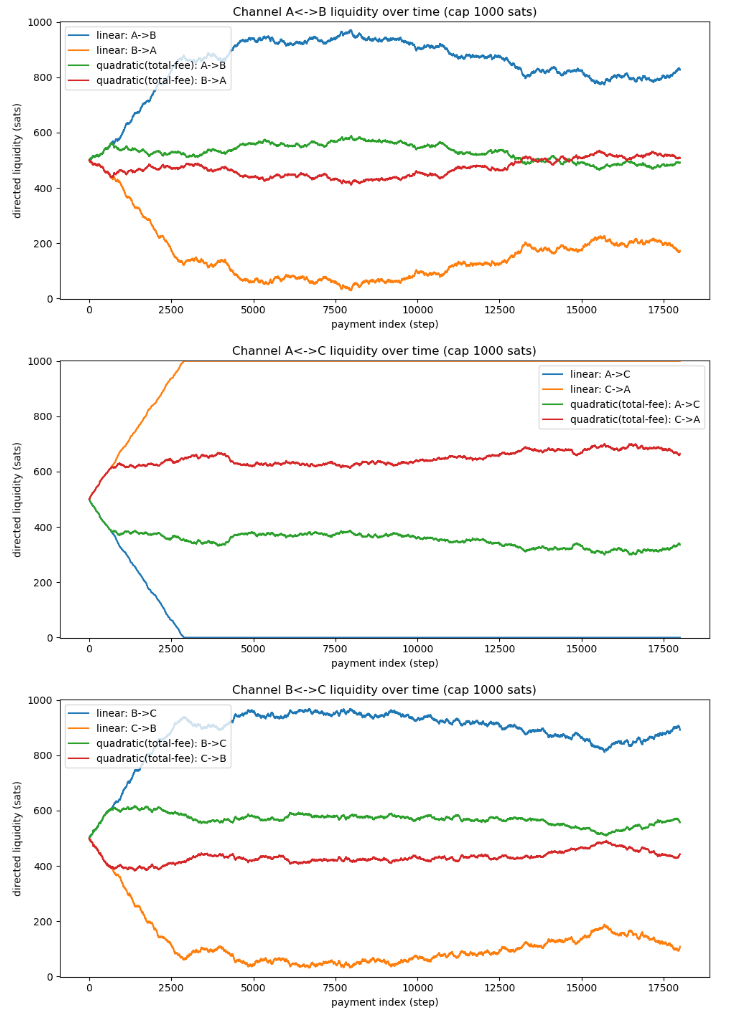}
\caption{On the time series we can see how the channel's liquidity depletes witha linear fee schedule but stays relatively balanced in the convex fee schedule.}
  \label{fig:convexTimeSeries}
\end{figure}

\begin{figure}[h]
\centering
\includegraphics[width=0.5\textwidth]{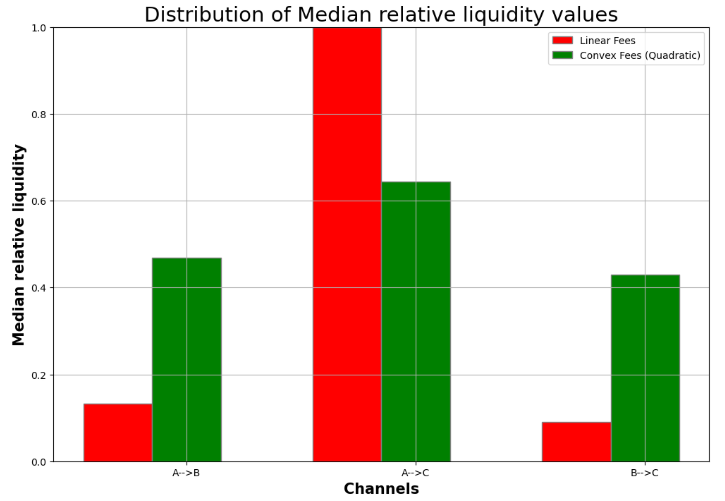}
\caption{The median relative liquidity summarizes the time series and shows how convex fee scheduls lead to more balanced liquidity on average.}
\label{fig:convexMedianLiquidity}
\end{figure}

In Figure \ref{fig:convexTimeSeries} we can see that the steady state of the linear fee protocol contains depleted channels. In contrast the steady state under a convex fee regime is within the interior of channels. Thus sufficient liquidity for feasibile payment requests is always  available. The time seriese are well summarized when looking at the median values from figure \ref{fig:convexMedianLiquidity}

\begin{figure}[h]
\centering
\includegraphics[width=0.5\textwidth]{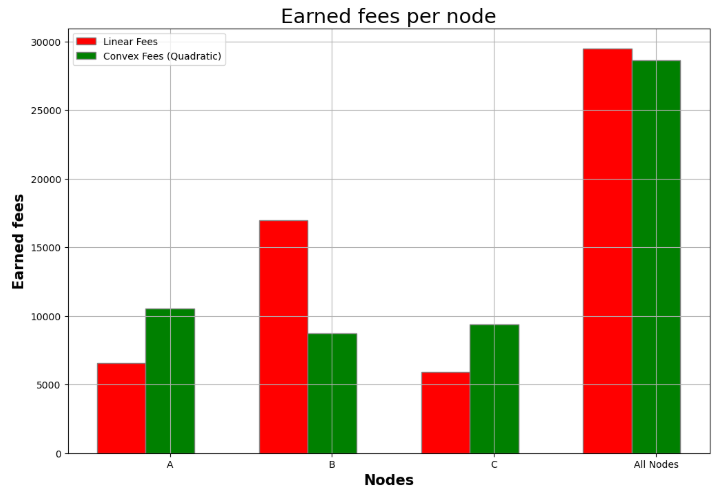}
\caption{Total network fees hardly change when going from linear to convex fees. However the earnings per node can change substatially. With a network wide convex fee schedule the earnings per node seem to equalize}
  \label{fig:convexFees}
\end{figure}

Finally in figure \ref{fig:convexFees} Despite the stark difference in depletion, aggregate routing fees over not change substantially across the two schedules. However it did change substentially for single nodes. Intuitively, convexity reshapes where flow runs (away from boundaries) rather than how much fee per unit is collected on average.

These observations support the view that convex, tiered fees act as flow control: They create interior stall points on cycles, preventing channels from hitting the boundary under persistent directional demand. However, fully utilization-responsive tiers complicate source routing: The effective end-to-end price depends on instantaneous, hidden liquidity $\lambda(e,\cdot)$ at each hop. Without additional mechanisms (coarse tier announcements with hysteresis, quote-based routing, or limited state disclosure), senders cannot price paths accurately. Exploring such protocol options lies beyond the scope of this experiment, but the qualitative stabilization result is robust across the small-network scenarios we tested.

\subsection{Optimal collaborative channel replenishment}
We have seen that for a given wealth distribution $\omega$ there may be several feasible liquidity distributions in the fiber $\pi^{-1}(\{\omega\})$. In a later chapter we will see that the mainly linear structure of the fee function and the economic rational behavior of sending nodes will lead to channel delpetion. This a burden for the network as feasible payments may take time and several attempts to be delivered which effectively reduces reliability and time to complete payments.

It would be desireable if channels could reprovide liquidity without doing expensive on chain transactions as those need to be saved to adress the infeasible payments.
Node operators often discuss the idea of circular rebalancings. However given the fact that the number of strict circulations are equivalent to the number of equivalent liquidity states we can ask if nodes may collaboratively select a circulation that is benefitial to their own liquidity desires and interests?
This can be done using the geometric properties:

The channel replenishment problem is equivalent to finding a circulation that shifts the liquidity to a more desireable state (e.g. less depletion). This problemen reduces to selecting an optimal element from $[\lambda]$. As demonstrated previously, feasible solutions can be obtained by solving a system of linear Diophantine equations subject to two constraint classes \textbf{Conservation of Liquidity} and \textbf{Conservation of Wealth}.
The solution space forms a convex polytope $P \subset \mathbb{Z}^{2m} \cap [0, c(u,v)]^{2m}$, which may be empty for certain $(\omega, G)$ pairs due to the intersection with the capacity hypercube.

Our optimization approach introduces a target liquidity state $x_0 \in \mathbb{Z}^{2m}$, representing the desired channel configuration. While we typically select $x_0$ where each node holds exactly half of each channel's capacity ($c(u,v)/2$), the framework accommodates arbitrary target states. The optimal replenishment problem then becomes:

Minimize $||x - x_0||_2$ subject to $x \in P \cap \mathbb{Z}^{2m}$.

Unfortunately it is known that optimizing a quadratic function over integers is np-hard \cite{delpia2014mixedintegerquadraticprogrammingnp}. However we present a very practical heuristic that delivers an exact solution that is close enough to the optimal solution. The important point is that we need an exact solution and not just an approximate solution. Thus we approach the problem in two phases: We start with a \textbf{Continuous relaxation} by computing $x_{\rho} = argmin\left(||x - x_0||_2\right)$ over $P \subset \mathbb{R}^{2m}$ using quadratic programming. In a second step we create a $\delta$ cube around $x_{\rho}$ in which we find some feasible \textbf{integer approximation}:

\[B_\delta(x_\rho) := \left\{x \in \mathbb{Z}^{2m} : ||x - x_\rho||_\infty \leq \delta\right\}\]

We suggest to choose $\delta = \left|\sqrt{\frac{||x_\rho - x_0)||_2}{m}))}+1\right|$ though one could search for $\delta$ with step size control. However our choice of $\delta$ ensures the neighborhood contains feasible integer solutions while maintaining computational tractability. 

We implemented the protocol and run it on a network with the topoology of a random graph with 100 nodes and 1556 channels. We run payments until the network went into a state with most channel's being depleted. In that state $11.31\%$ of all channels have between have liquidity between $0.4$ and $0.6$
and $41.52\%$ of all channels have between have liquidity between $0.1$ and $0.9$.

\begin{figure}[h]
\centering
 \includegraphics[width=0.41\textwidth]{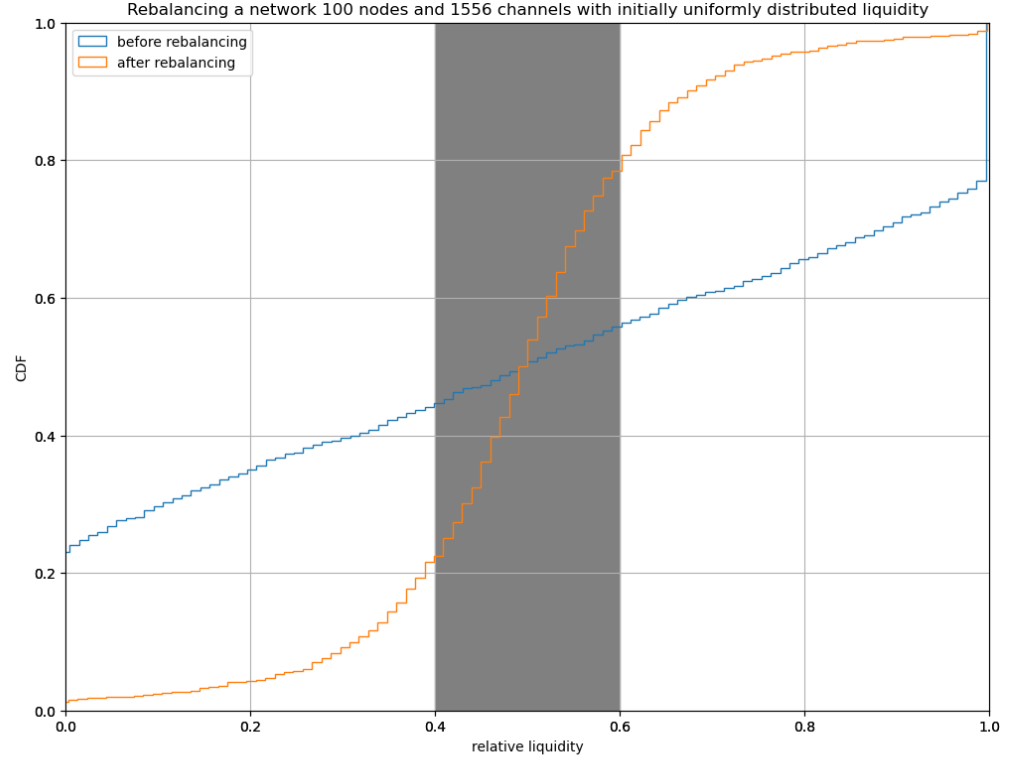}
\caption{We can see the stark shift of the shape of the liquidity distribution when using the geometric channel replenishment protocol}
\label{fig:replenishment}
\end{figure}

Figure \ref{fig:replenishment} indicates that after running the proposed protocool to find the ciruclation that moves the liquidity - given the wealth distribution closest to the center of the liquidity poloytope we can see that $56.30\%$ of all channels have their liquidity between $0.4$ and $0.6$ and $95.24\%$ of all channels have their liquidity values between $0.1$ and $0.9$.
The shape of the distribution went from the typical bimodal distribution to a normal distribution with moste liquidity centered around the middle of the channel. This would allow for quick path selection of feasible payments.

In order to achieve this rebalancing $33.23\%$ of the networks capacity had to be allocated to htlcs and rebalanced. We have not studied how this optimal geometric approach outperforms greedy heuristics\cite{pickhardt2019imbalancemeasureproactivechannel} that we have proposed in the past.

As we took an exact integer approximation within the $\delta$-environment of the real solution $x_\rho$ it is interesting to compare how far the integer and real solutions are away from the desired state $x_0$. The optimal solution over the real numbers has in this case a distance of $91092.32$ where the non optimal but feasible integer solution had a distance of $91096.86$. We note that the difference seems neglectible though a formal proof of an approximation guarantee of the proposed method would be desireable. 

It is important to note that a chanenl replenishment protocol like this yields severe practical issues:
\begin{enumerate}
\item A centralized coordinator is needed.
\item Nodes need to share their wealth information.
\item If individual target states are allowed (instead of just aiming for a 50/50 channel as we did in our simulation) peers must not contradict each other.
  \item It may be hard to run the protocol while the network is active.
\end{enumerate}
For those reasons we assume that this result is just of theoretical interest or best to be implemented by a lighting network service provider who could offer this kind of replenishment protocol as a service to the customers of that provider.

\section{Discussion and Conclusion}

Our geometric lens makes two shortcomings of today’s two–party, source–routed LN explicit. First, many desired target wealth vectors lie outside the feasible region $W_G$ because sparse topologies have narrow cut intervals and the conservation of liquidity heavily constraints the ILP that defines $W_G$. This inflates the expected infeasible rate $\rho$ and, by $S=\zeta/\rho$, caps the sustainable off-chain payment bandwidth. Second, under linear, asymmetric fees the dynamics inside a wealth fiber tend to deplete channels. This reduces the reliability of otherwise feasible payments. Given the current incentives in the network depletion is an emergent phenomenon even in a circular economy when all payments are net-zero. The existence of infeasible payments and depletion are consequences of the geometry of cuts and cycles under the given current protocol design. They are neither software bugs nor necessarily the consequence of poor node operator behavior. 

The presented theory is constructive as it identifies tunable levers:
\begin{itemize}
  \item \textbf{Grow $W_G$ (reduce $\rho$).} Favor multi-party primitives like coin-pools or channel factories that widen cuts in expectation. Improve the inter-pool connectivity and place capital with cut widths in mind.
  \item \textbf{Raise effective $\zeta$.} Batch settlements via factories or other means like splices. Move recurring flows on-chain in larger units andexplore designs that amortize on-chain bandwidth across many off-chain transfers. Also here multi-party primitives seem desireable.
  \item \textbf{Tame depletion.} (i) Symmetric fees remove directional imbalances (simple, though likely unpopular) (ii) Convex fees create interior stall points on cycles (effective locally, but they complicate source routing unless channel state is broadcasted (iii) Coordinated replenishment (choose an optimal circulation within a fiber) recenters liquidity without changing wealth.
\end{itemize}

To isolate first principles, we abstracted from base fees, upfront fees, CLTV deltas, HTLC count/size limits, liquidity ads, and jamming countermeasures. These details matter for engineering targets and constants. These details - while important in practise - do not alter the structural facts that (i) feasibility is governed by cut intervals (hence $W_G$) and (ii) linear asymmetric fees drive cycle-wise boundary optima (hence depletion).

A key open design choice is whether multi-party channels should primarily serve as channel factories (periodically batching updates and effectively increasing $\zeta$), or as persistent payment venues that also enlarge $W_G$ between batches (e.g., cross-pool swaps). The former is operationally simpler and immediately helps $S$ via $\zeta$; the latter promises larger feasible regions and lower $\rho$ but requires careful incentives, privacy, and failure handling. Our results suggest both directions are fruitful. Future research is needed to decide the best mix which likely depends on demand patterns and governance constraints.

The geometric viewpoint turns system evolution into measurable goals: track the infeasible payment rate $\rho$, track depletion, and act on the levers above. In practice, widen cuts with multi-party primitives (channel factories/coin-pools), keep operation near interior fibers with convex pricing or coordinated replenishment, and amortize on-chain costs with batching. It is up to the community to judge how severe these theoretical limits are in practice and which levers to integrate into the protocol. For depletion, some levers seem technically easier (symmetric fees) and others harder (utilization-responsive fees under source routing or off-chain replenishment). By contrast, infeasibility appears to require deeper protocol upgrades: with two-party channels, $W_G$ is likely too small and available on-chain bandwidth insufficient for scale. With targeted support, however, payment channel networks can spend more time inside the feasible region and deliver cheaper, more reliable payments.

The machinery we propose to widen feasibility is necessary precisely because we insist on a trust-minimised, non-custodial system. As shown in Section~\ref{sec:credit}, allowing credit would make many liquidity constraints evaporate at the cost of the very property that motivates this line of work. As Nakamoto noted:
\begin{displayquote}
\emph{“The root problem with conventional currency is all the trust that’s required to make it work. The central bank must be trusted not to debase the currency, but the history of fiat currencies is full of breaches of that trust.”}\footnote{Satoshi Nakamoto (2009)\url{https://satoshi.nakamotoinstitute.org/posts/p2pfoundation/1/}} 
\end{displayquote}
Our results delineate the engineering frontier within those constraints and show what is possible to be achieved: Channel factories and coin-pools measurably enlarge $W_G$ without reintroducing trusted intermediaries. The remaining task is to scale these levers into deployable, privacy-preserving protocols.

\section{Acknowledgements}
This work was funded by Open Sats Inc. and various Patreons. I would like to thank the reviewers of earlier versions. In particular Stefan Richter, Christian Decker, Christian Kuemmerle, Anastasios Sidiropoulos (who pointed out an error in my initial description of the wealth distributions), Elias Roher and bitromortac. Thanks to Rene Treffer for sharing an example with me that indicated how dropping the assumption that liquidity in channels is distributed independently can be used to show how bimodal liquidity distributions emerge.
I am not a native English speaker. The first full draft was copy-edited with the help of ChatGPT. I reviewed and revised all suggested changes, and the original text is preserved as comments in the LaTeX source. In subsequent revisions and extensions I adopted a more vibe-writing approach, so sentence-level attribution of AI-assisted phrasing is no longer practical. Conceptually, while I was working to formalize results on the feasibility region of multi-party channels, the AI—building on my initial ideas suggested the perspective of analyzing expected cut widths, which informed parts of this manuscript.
%Furthermore since I am not a native english speaker I passed the first version of the manuscript through chatGPT for copy editing and revised its output. My origninal manuscript is as comments in the latex files. For the second iteration and extension of the original script I switched more to a vibe-writing style so it is harder to mark the AI generated parts of this manuscript explicitely. It is noteworthy while I was struggeling to formalize proofs about multiparty channels the AI suggested based on my ideas to study expected cut widths.
\bibliography{references}

@article{burchert2018scalable,
  title={Scalable funding of bitcoin micropayment channel networks},
  author={Burchert, Conrad and Decker, Christian and Wattenhofer, Roger},
  journal={Royal Society open science},
  volume={5},
  number={8},
  pages={180089},
  year={2018},
  publisher={The Royal Society}
}

@misc{Pickhardt2025,
title={Ark as a channel factory: Compressed liquidity management for improved payment feasibility},
url={https://delvingbitcoin.org/t/ark-as-a-channel-factory-compressed-liquidity-management-for-improved-payment-feasibility/2179},
journal={Delving Bitcoin},
author={Pickhardt, Rene},
year={2025},
month={Dec}}

@article{decker2018eltoo,
  title={eltoo: A simple layer2 protocol for bitcoin},
  author={Decker, Christian and Russell, Rusty and Osuntokun, Olaoluwa},
  journal={White paper: https://blockstream. com/eltoo. pdf},
  year={2018}
}

@misc{pickhardt2019imbalancemeasureproactivechannel,
      title={Imbalance measure and proactive channel rebalancing algorithm for the Lightning Network}, 
      author={Rene Pickhardt and Mariusz Nowostawski},
      year={2019},
      eprint={1912.09555},
      archivePrefix={arXiv},
      primaryClass={cs.SI},
      url={https://arxiv.org/abs/1912.09555}, 
}

@misc{Neigut2022,
title={Fee Ratecards (your gateway to negativity)},
url={https://diyhpl.us/\~bryan/irc/bitcoin/bitcoin-dev/linuxfoundation-pipermail/lightning-dev/2022-September/003685.txt},
journal={Diyhpluswiki},
author={Neigut, Lisa},
year={2022},
month={Sep}}

@ARTICLE{elias1956note,
  author={Elias, P. and Feinstein, A. and Shannon, C.},
  journal={IRE Transactions on Information Theory}, 
  title={A note on the maximum flow through a network}, 
  year={1956},
  volume={2},
  number={4},
  pages={117-119},
  doi={10.1109/TIT.1956.1056816}}

@article{FordFulkerson1956,
title={Maximal Flow Through a Network},
volume={8},
DOI={10.4153/CJM-1956-045-5},
journal={Canadian Journal of Mathematics},
author={Ford, L. R. and Fulkerson, D. R.},
year={1956},
pages={399–404}}

@misc{delpia2014mixedintegerquadraticprogrammingnp,
      title={Mixed-integer Quadratic Programming is in NP}, 
      author={Alberto Del Pia and Santanu S. Dey and Marco Molinaro},
      year={2014},
      eprint={1407.4798},
      archivePrefix={arXiv},
      primaryClass={cs.DM},
      url={https://arxiv.org/abs/1407.4798}, 
}

@book{steele2004cauchy,
  title={The Cauchy-Schwarz master class: an introduction to the art of mathematical inequalities},
  author={Steele, J Michael},
  year={2004},
  publisher={Cambridge University Press}
}

@inproceedings{harris2020flood,
  title={Flood \& loot: A systemic attack on the lightning network},
  author={Harris, Jona and Zohar, Aviv},
  booktitle={Proceedings of the 2nd ACM Conference on Advances in Financial Technologies},
  pages={202--213},
  year={2020}
}

@article{shikhelman2022unjamming,
  title={Unjamming lightning: A systematic approach},
  author={Shikhelman, Clara and Tikhomirov, Sergei},
  journal={Cryptology ePrint Archive},
  year={2022}
}

@inproceedings{tochner2020route,
  title={Route hijacking and dos in off-chain networks},
  author={Tochner, Saar and Zohar, Aviv and Schmid, Stefan},
  booktitle={Proceedings of the 2nd ACM Conference on Advances in Financial Technologies},
  pages={228--240},
  year={2020}
}

@article{pickhardt2021security,
  title={Security and privacy of lightning network payments with uncertain channel balances},
  author={Pickhardt, Rene and Tikhomirov, Sergei and Biryukov, Alex and Nowostawski, Mariusz},
  journal={arXiv preprint arXiv:2103.08576},
  year={2021}
}

@article{pickhardt2021optimally,
  title={Optimally reliable \& cheap payment flows on the lightning network},
  author={Pickhardt, Rene and Richter, Stefan},
  journal={arXiv preprint arXiv:2107.05322},
  year={2021}
}

@online{bitromortac2024blazing,
  title={Blazing the Trails: Improving LND Pathfinding Reliability},
  author={bitromortac},
  year={2024},
  url={\url{https://lightning.engineering/posts/2024-05-23-pathfinding-2/}},
  urldate={2024-03-24},
  publisher={Lightning Labs}
}

@online{pickhardt2022valves,
  title={The power of valves for better flow control, improved reliability and lower expected payment failure rates on the Lightning Network},
  author={Pickhardt, Rene},
  year={2022},
  url={\url{https://blog.bitmex.com/the-power-of-htlc_maximum_msat-as-a-control-valve-for-better-flow-control-improved-reliability-and-lower-expected-payment-failure-rates-on-the-lightning-network/}},
  urldate={2022-09-21},
  publisher={Bitmex}
}

@article{rossi2024channel,
  title={Channel Balance Interpolation in the Lightning Network via Machine Learning},
  author={Rossi, Emanuele and Singh, Vikash and others},
  journal={arXiv preprint arXiv:2405.12087},
  year={2024}
}

@inproceedings{zhang2023rethinking,
  title={Rethinking Incentive in Payment Channel Networks},
  author={Zhang, Yunqi and Venkatakrishnan, Shaileshh Bojja},
  booktitle={2023 IEEE 43rd International Conference on Distributed Computing Systems Workshops (ICDCSW)},
  pages={61--66},
  year={2023},
  organization={IEEE}
}

@inproceedings{ren2018optimal,
  title={Optimal fee structure for efficient lightning networks},
  author={Ren, Alvin Heng Jun and Feng, Ling and Cheong, Siew Ann and Goh, Rick Siow Mong},
  booktitle={2018 IEEE 24th International Conference on Parallel and Distributed Systems (ICPADS)},
  pages={980--985},
  year={2018},
  organization={IEEE}
}

@article{alscher2023price,
  title={Price of Anarchy in the Lightning Network},
  year={2023},
  author={Alscher, Sebastian}
}

@inproceedings{pickhardt2020imbalance,
  title={Imbalance measure and proactive channel rebalancing algorithm for the lightning network},
  author={Pickhardt, Rene and Nowostawski, Mariusz},
  booktitle={2020 IEEE International Conference on Blockchain and Cryptocurrency (ICBC)},
  pages={1--5},
  year={2020},
  organization={IEEE}
}

@inproceedings{dandekar2011liquidity,
  title={Liquidity in credit networks: A little trust goes a long way},
  author={Dandekar, Pranav and Goel, Ashish and Govindan, Ramesh and Post, Ian},
  booktitle={Proceedings of the 12th ACM conference on Electronic commerce},
  pages={147--156},
  year={2011}
}

@online{piatkivskyi2018rebalancing,
title={Rebalancing argument},
author={Piatkivskyi, Dmytro},
url={\url{https://lists.linuxfoundation.org/pipermail/lightning-dev/2018-July/001329.html}},
year={2018}
}

@online{guidi2019steadystate,
title={Paper - Modeling a Steady-State Lightning Network Economy},
author={Guidi, Gregorio},
url={\url{https://lists.linuxfoundation.org/pipermail/lightning-dev/2019-August/002115.html}},
year={2019}
}

@article{de2004effective,
  title={Effective lattice point counting in rational convex polytopes},
  author={De Loera, Jes{\'u}s A and Hemmecke, Raymond and Tauzer, Jeremiah and Yoshida, Ruriko},
  journal={Journal of symbolic computation},
  volume={38},
  number={4},
  pages={1273--1302},
  year={2004},
  publisher={Elsevier}
}

@article{barvinok1994polynomial,
  title={A polynomial time algorithm for counting integral points in polyhedra when the dimension is fixed},
  author={Barvinok, Alexander I},
  journal={Mathematics of Operations Research},
  volume={19},
  number={4},
  pages={769--779},
  year={1994},
  publisher={INFORMS}
}

@article{ehrhart1962polyhedra,
  title={On homothetic rational poly{e}dras {\`a} n dimensions},
  author={Ehrhart, Eugene},
  journal={CR Acad. Sci. Paris},
  volume={254},
  pages={616},
  year={1962}
}

@inproceedings{chen2022maximum,
  title={Maximum flow and minimum-cost flow in almost-linear time},
  author={Chen, Li and Kyng, Rasmus and Liu, Yang P and Peng, Richard and Gutenberg, Maximilian Probst and Sachdeva, Sushant},
  booktitle={2022 IEEE 63rd Annual Symposium on Foundations of Computer Science (FOCS)},
  pages={612--623},
  year={2022},
  organization={IEEE}
}

@article{chen2023almost,
  title={Almost-Linear-Time Algorithms for Maximum Flow and Minimum-Cost Flow},
  author={Chen, Li and Kyng, Rasmus and Liu, Yang P and Peng, Richard and Gutenberg, Maximilian Probst and Sachdeva, Sushant},
  journal={Communications of the ACM},
  volume={66},
  number={12},
  pages={85--92},
  year={2023},
  publisher={ACM New York, NY, USA}
}

@inproceedings{van2023deterministic,
  title={A deterministic almost-linear time algorithm for minimum-cost flow},
  author={Van Den Brand, Jan and Chen, Li and Peng, Richard and Kyng, Rasmus and Liu, Yang P and Gutenberg, Maximilian Probst and Sachdeva, Sushant and Sidford, Aaron},
  booktitle={2023 IEEE 64th Annual Symposium on Foundations of Computer Science (FOCS)},
  pages={503--514},
  year={2023},
  organization={IEEE}
}

@inproceedings{griffin2020generating,
  title={Generating utilization vectors for the systematic evaluation of schedulability tests},
  author={Griffin, David and Bate, Iain and Davis, Robert I},
  booktitle={2020 IEEE Real-Time Systems Symposium (RTSS)},
  pages={76--88},
  year={2020},
  organization={IEEE}
}

@article{tikhomirov2020probing,
 title={Probing channel balances in the lightning network},
  author={Tikhomirov, Sergei and Pickhardt, Rene and Biryukov, Alex and Nowostawski, Mariusz},
  journal={arXiv preprint arXiv:2004.00333},
  year={2020}
}

@misc{poon2016bitcoin,
  title={The bitcoin lightning network: Scalable off-chain instant payments},
  author={Poon, Joseph and Dryja, Thaddeus},
  year={2016}
}
\bibliographystyle{plain}

\appendix

  \section{Low Dimensional Example}
\label{sec:paymentsExample}
We have already seen that for different topologies the feasible region and thus $r(G)$.
For our example we set:
\begin{equation*}
  \begin{split}
    V   & =\{x,y,z\} \\
    E=\{e=(x,y),f & =(y,z), g=(x,z)\}\\
    c_e=3, c_f & =7, c_g=11
  \end{split}
\end{equation*}

We want to see if the wealth vector $w=(5,6,10)$ is feasible in this network.
thus we construct our system of linear equations as.

We have $m$ constraints related to conservation of liquidity.
\begin{align*}
  e_x + e_y  & = c_e \\
  f_y + f_z & = c_f \\
  g_x + g_z & = c_g \\
\end{align*}
and $n$ constraints related to the wealth vector.
\begin{align*}
  e_x +  g_x & = w_x \\
  e_y +  f_y & = w_y \\
  f_z +  g_z & = w_z
\end{align*}
Together this results in the following system of linear equations.
\[
  \left[
    \begin{array}{cccccc}
      1 & 1 & 0 & 0 & 0 & 0 \\
      0 & 0 & 1 & 1 & 0 & 0 \\
      0 & 0 & 0 & 0 & 1 & 1 \\
      1 & 0 & 0 & 0 & 1 & 0 \\
      0 & 1 & 1 & 0 & 0 & 0 \\
      0 & 0 & 0 & 1 & 0 & 1 
    \end{array}
    \right]
\begin{bmatrix}
  e_{x}  \\
  e_{y}  \\
  f_{y}  \\
  f_{z}  \\
  g_{x}  \\
  g_{z}  \\
\end{bmatrix} 
  = \begin{bmatrix}
  c_e\\
  c_f\\
  c_g\\
  w_x\\
  w_y\\
  w_z
\end{bmatrix}
\]
Again we ca replace the variables $e_y, f_z$ and $g_z$ with $e_x, f_y$ and $g_x$ respectively. 
\[
  \left[
    \begin{array}{ccc}
      1 & 0 & 1 \\
      - 1& 1 & 0 \\
      0 & -1 & -1 
    \end{array}
    \right]
\begin{bmatrix}
  e_{x}  \\
  f_{y}  \\
  g_{x}  \\
\end{bmatrix} 
  = \begin{bmatrix}
  w_x\\
  w_y\\
  w_z
  \end{bmatrix}
  - \begin{bmatrix}
  0\\
  c_e\\
  c_f + c_g
  \end{bmatrix}
\]

The solution space is $1$ dimensional and can be parmeterized via:

$$\sigma = \{\begin{pmatrix}0\\w_y-c_e\\w_x\end{pmatrix} + t\cdot\begin{pmatrix}1\\1\\-1\end{pmatrix}\in\mathbb{Z}^3| t\in\mathbb{Z}\}$$

Plugging in the wealth vector $w=(5,6,10)^t$ and the capacities we get:

$$\sigma = \{\begin{pmatrix}0\\3\\5\end{pmatrix} + t\cdot\begin{pmatrix}1\\1\\-1\end{pmatrix}\in\mathbb{Z}^3| t\in\mathbb{Z}\}$$

\begin{figure}[h]
  \label{fig:preimage}
\centering
\includegraphics[width=0.5\textwidth]{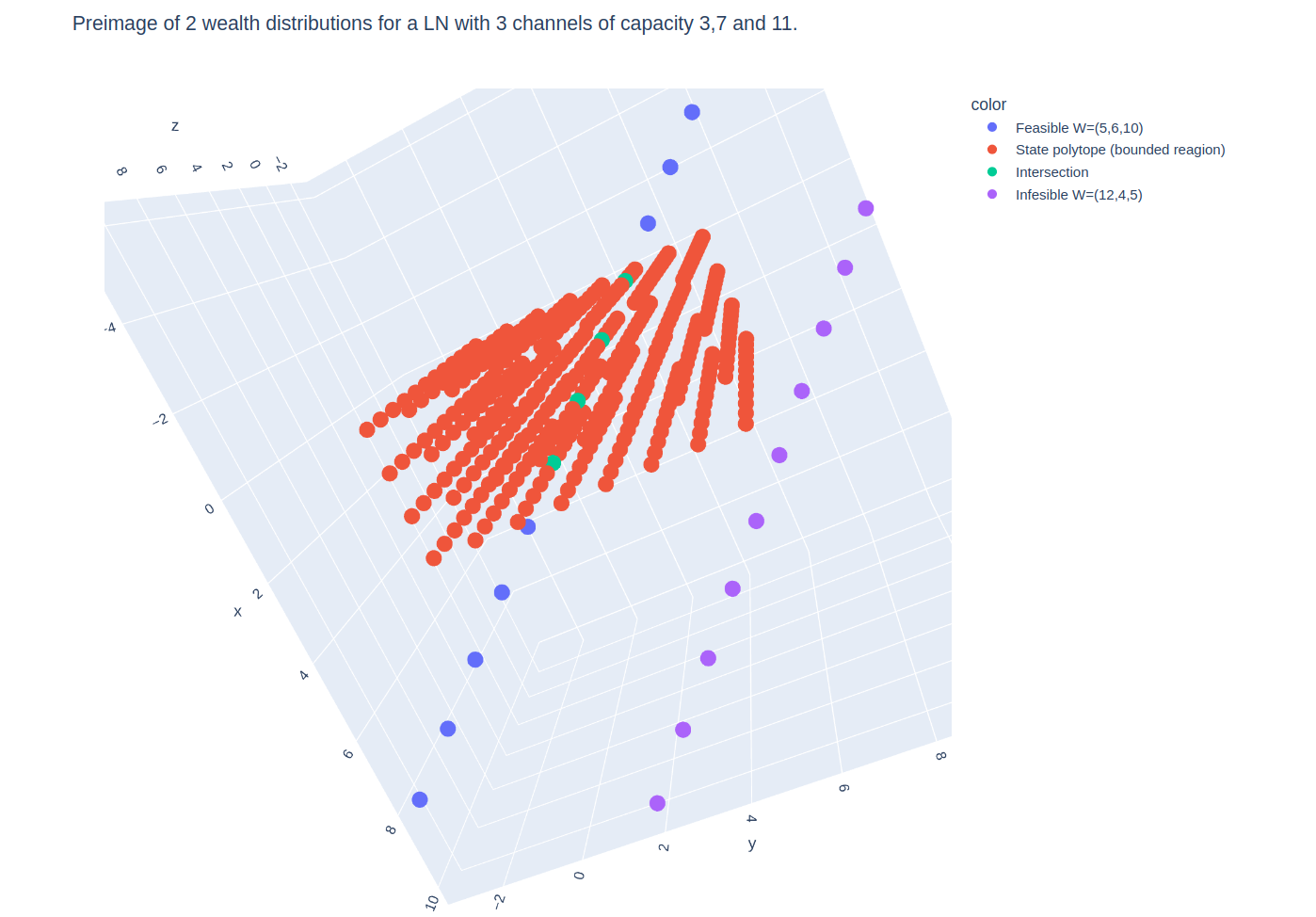}
\caption{We see the preimage of $\pi|_{\mathbb{Z}^m}$ for $2$ wealth distributions. One does have an empty intersection with the bounding box of the state polytope and is thus infeasible on the example network.}
\end{figure}

Finally we racall that our state space $S=\{0,\dots,3\}\times\{0,\dots,7\}\times\{0,\dots,11\}$.
We define $F = \sigma\cap S$ to be the set of feasible solutions.
We see that the parameter $t$ cannot be negative. Thus we get the following table of feasible solutions.

%via: https://www.tablesgenerator.com/
\begin{table}[]
\begin{tabular}{|l|c|c|c|c|}
\hline
$t$                                                & $0\leq e_x\leq 3$   & $0\leq f_y\leq 7$   & $0\leq g_x\leq 11$   & \textbf{feasable}        \\ \hline
\rowcolor[HTML]{34FF34} 
\cellcolor[HTML]{FE0000}-1                       & \cellcolor[HTML]{FE0000}-1  & 2                        & 6                        & \cellcolor[HTML]{FE0000}N \\ \hline
\rowcolor[HTML]{34FF34} 
{\color[HTML]{000000} 0}                         & {\color[HTML]{000000} 0}    & {\color[HTML]{000000} 3} & {\color[HTML]{000000} 5} & {\color[HTML]{000000} Y}  \\ \hline
\rowcolor[HTML]{34FF34} 
1                                                & 1                           & 4                        & 4                        & Y                         \\ \hline
\rowcolor[HTML]{34FF34} 
2                                                & 2                           & 5                        & 3                        & Y                         \\ \hline
\rowcolor[HTML]{34FF34} 
3                                                & 3                           & 6                        & 2                        & Y                         \\ \hline
\rowcolor[HTML]{34FF34} 
\cellcolor[HTML]{FE0000}{\color[HTML]{000000} 4} &  \cellcolor[HTML]{FE0000}4  & 7                        & 1                        & \cellcolor[HTML]{FE0000}N \\ \hline
\end{tabular}
\end{table}

If $y$ wants to make a payment of $2$ coins to $z$ we take the wealth vector $w=(5,6,10)^t$ and compute:

\[w'=\begin{pmatrix}5\\6\\10\end{pmatrix}-2\cdot b_y + 2\cdot b_z = \begin{pmatrix}5\\4\\12\end{pmatrix}\]

  We get:
\[\sigma' = \{\begin{pmatrix}0\\1\\5\end{pmatrix}\ + t\cdot \begin{pmatrix}1\\1\\-1\end{pmatrix}\}\]

  For $t=0$ we have $(0,1,5)^t\in S$.
  Thus the wealth distribution $w'\in W$ and the payment was feasible.

  Assuming we have the wealth distribution $w'$ and user $z$ wants to pay $10$ coins to $x$ then we had the wealth distribution $\omega=(15,4,2)$.
  As we have alredy seen in figure \ref{fig:preimage} that the solution space of $\omega$ does have an empty intersection with $S$.
  Thus the payment of $10$ coins is infeasible.

\end {document}